\RequirePackage[l2tabu,orthodox]{nag}		% for deprecated commands
\PassOptionsToPackage{dvipsnames,svgnames}{xcolor}		% for xcolor conflict
\PassOptionsToPackage{inline,shortlabels}{enumitem}		% for enumitem conflict
\PassOptionsToPackage{numbers,sort&compress}{natbib}		% for natbib conflict
\documentclass[full,12pt]{colt2021} % 

\usepackage{amsmath}		% for AMS macros
\usepackage{amssymb}		% for AMS symbols
\usepackage{amsfonts}		% for AMS fonts
\usepackage{amsthm}		% for theorems

\usepackage{mathtools}		% for advanced math

\mathtoolsset{%
}

\usepackage[utf8]{inputenc}		% for source encoding
\usepackage[T1]{fontenc}		% for font encoding

\usepackage[%		% for math font selection
cal=cm,
]
{mathalfa}

\usepackage{dsfont}		% for blackboard bold font
\usepackage[scale=.96]{inconsolata}

\usepackage[proportional,tabular,lining,sf,mono=false]{libertine}
\usepackage{times}
\usepackage{acronym}		% for acronyms
\renewcommand*{\aclabelfont}[1]{\acsfont{#1}}		% for acronym label font
\newcommand{\acli}[1]{\emph{\acl{#1}}}		% for italicized acro
\newcommand{\acdef}[1]{\emph{\acl{#1}} \textup{(\acs{#1})}\acused{#1}}		% for acro def
\newcommand{\acdefp}[1]{\emph{\aclp{#1}} \textup{(\acsp{#1})}\acused{#1}}	% for acro def (plural)

\usepackage[labelfont={bf,small},labelsep=colon,font=small]{caption}	% for caption control
\colorlet{MyRed}{FireBrick}
\colorlet{MyGreen}{DarkGreen}
\colorlet{MyBlue}{MediumBlue}

\newcommand{\afterhead}{.\;}		% for changing headings
\newcommand{\para}[1]{\medskip\paragraph{\textbf{#1\afterhead}}}

\usepackage{latexsym}		% for symbols

\usepackage{pifont}		% for dingbats
\usepackage{tikz}		% for figures
\usetikzlibrary{calc,patterns}		% for basic tikz figures

\usepackage{array}		% for flexible tables
\usepackage{booktabs}		% for better tables
\usepackage[inline,shortlabels]{enumitem}		% for lists
\setenumerate{itemsep=0pt,topsep=\smallskipamount}		% for list controls

\usepackage[kerning=true]{microtype}		% for microtypography

\usepackage{xspace}		% for flexible spaces

\bibpunct[, ]{[}{]}{,}{n}{,}{,}

\newcommand{\citef}[1]{\citeauthor{#1} \citep{#1}}
\usepackage{hyperref}
\hypersetup{
colorlinks=true,
linktocpage=true,
pdfstartview=FitH,
breaklinks=true,
pdfpagemode=UseNone,
pageanchor=true,
pdfpagemode=UseOutlines,
plainpages=false,
bookmarksnumbered,
bookmarksopen=false,
bookmarksopenlevel=1,
hypertexnames=true,
pdfhighlight=/O,
urlcolor=MyBlue,linkcolor=MyBlue,citecolor=MyBlue,	% for on-screen
pdftitle={},
pdfauthor={},
pdfsubject={},
pdfkeywords={},
pdfcreator={pdfLaTeX},
pdfproducer={LaTeX with hyperref}
}

\usepackage[sort&compress,capitalize,nameinlink]{cleveref}		% for cleveref formatting
\crefname{assumption}{Assumption}{Assumptions}
\usepackage{algorithm}		% for algorithm environments
\usepackage{algpseudocode}		% for algorithm macros
\usepackage{thmtools}		% for theorem tools
\usepackage{thm-restate}		% for restating theorems

\newtheorem{theorem}{Theorem}		% for theorems
\newtheorem{corollary}{Corollary}		% for corollaries
\newtheorem{lemma}{Lemma}		% for lemmas
\newtheorem{proposition}{Proposition}		% for propositions
\newtheorem{fact}{Fact}
\newtheorem{model}{Model}

\newtheorem*{corollary*}{Corollary}		% for corollaries (unnumbered)

\newtheorem{definition}{Definition}		% for definitions
\newtheorem{example}{Example}		% for examples
\newtheorem*{definition*}{Definition}		% for definitions (unnumbered)
\newtheorem*{assumption*}{Assumptions}		% for assumptions (unnumbered)
\newtheorem*{example*}{Example}		% for examples (unnumbered)

\theoremstyle{remark}
\newtheorem{remark}{Remark}		% for remarks

\newtheorem*{remark*}{Remark}		% for remarks (unnumbered)

\newcounter{proofpart}

\usepackage[showdeletions]{color-edits}		% for editing macros / use [suppress] for final
\usepackage[normalem]{ulem}		% for strikeout text
\newcommand{\debug}[1]{#1}		% for removing macro coloring

\newcommand{\newmacro}[2]{\newcommand{#1}{\debug{#2}}}		% for shorthand definitions
\newcommand{\newop}[2]{\DeclareMathOperator{#1}{\debug{#2}}}		% for shorthand definitions

\DeclarePairedDelimiter{\braces}{\{}{\}}		% for braces
\DeclarePairedDelimiter{\bracks}{[}{]}		% for brackets
\DeclarePairedDelimiter{\parens}{(}{)}		% for parentheses

\DeclarePairedDelimiter{\abs}{\lvert}{\rvert}		% for absolute value
\DeclarePairedDelimiterX{\setdef}[2]{\{}{\}}{#1:#2}		% for set builder notation
\DeclarePairedDelimiterXPP{\exclude}[1]{\mathopen{}\setminus}{\{}{\}}{}{#1}

\newcommand{\R}{\mathbb{R}}		% for reals
\DeclareMathOperator*{\argmax}{arg\,max}		% for argmax
\DeclareMathOperator*{\argmin}{arg\,min}		% for argmin
\DeclareMathOperator*{\intersect}{\bigcap}		% for intersections
\DeclareMathOperator*{\union}{\bigcup}		% for unions

\DeclareMathOperator{\bigoh}{\mathcal{O}}		% for Landau O
\DeclareMathOperator{\dist}{dist}		% for distance
\DeclareMathOperator{\dom}{dom}		% for domain
\DeclareMathOperator{\im}{im}		% for image
\DeclareMathOperator{\one}{\mathds{1}}		% for indicator
\DeclareMathOperator{\relint}{ri}		% for relative interior
\DeclareMathOperator{\supp}{supp}		% for support
\newcommand{\cf}{cf.\xspace}		% for consistency
\newcommand{\eg}{e.g.,\xspace}		% for consistency
\newcommand{\ie}{i.e.,\xspace}		% for consistency
\newcommand{\textpar}[1]{\textup(#1\textup)}		% for upshape parentheses

\newcommand{\alt}[1]{#1'}		% for alternate version
\newcommand{\altalt}[1]{#1''}		% for second alternate
\newmacro{\dd}{\:d}		% for integration
\newcommand{\eps}{\varepsilon}		% for better epsilon
\newmacro{\const}{c}		% for generic constant
\newmacro{\coef}{\lambda}		% for generic coefficient
\newmacro{\param}{\theta}		% for parameter
\newmacro{\params}{\Theta}		% for set of parameters

\newmacro{\pexp}{p}		% for first exponent
\newmacro{\qexp}{q}		% for second exponent
\newmacro{\rexp}{r}		% for third exponent

\newmacro{\beforestart}{0}		% for before start index
\newmacro{\start}{0}		% for start index
\newmacro{\afterstart}{1}		% for second index
\newmacro{\running}{\start,\afterstart,\dotsc}		% for running

\newmacro{\run}{n}		% for main sequence index
\newmacro{\runalt}{k}		% for alternate index
\newmacro{\runaltalt}{m}		% for second alternate
\newmacro{\nRuns}{T}		% for total number of runs
\newmacro{\runs}{\mathcal{\nRuns}}		% for set of indices

\newmacro{\state}{X}		% for main state
\newmacro{\statealt}{Y}		% for alternate state
\newmacro{\statealtalt}{Z}		% for second alternate

\newcommand{\curr}[1][\state]{\debug{#1}_{\run}}		% for current iterate (X by default)
\newmacro{\tstart}{0}		% for time start
\newmacro{\timealt}{s}		% for dummy continuous time
\newmacro{\horizon}{T}		% for horizon

\newmacro{\traj}{x}		% for trajectory
\newmacro{\trajalt}{y}		% for alternate trajectory
\newmacro{\trajaltalt}{z}		% for second alternate

\newmacro{\flowmap}{\Theta}		% for (semi)flow map
\DeclarePairedDelimiterXPP{\flowof}[2]{\flowmap_{#1}}{(}{)}{}{#2}		% for flow

\newop{\Nash}{NE}		% for Nash equilibria
\newop{\CE}{CE}		% for correlated equilibria
\newop{\CCE}{CCE}		% for Hannan set
\newop{\NI}{NI}		% for Nikaido-Isoda function

\newop{\brep}{br}		% for best responses
\newop{\reg}{Reg}		% for regret
\newop{\preg}{\overline{Reg}}		% for pseudo-regret
\newop{\val}{val}		% for value function

\newmacro{\play}{i}		% for player index
\newmacro{\playalt}{j}		% for alternate player index
\newmacro{\playaltalt}{k}		% for second alternate
\newmacro{\nPlayers}{N}		% for number of players
\newmacro{\players}{\mathcal{\nPlayers}}		% for set of players

\newmacro{\pure}{\alpha}		% for pure strategy
\newmacro{\purealt}{\beta}		% for alternate pure strategy
\newmacro{\purealtalt}{\gamma}		% for second alternate
\newmacro{\nPures}{A}		% for number of pure strategies
\newmacro{\pures}{\mathcal{\nPures}}		% for set of pure strategies

\newmacro{\strat}{x}		% for mixed strategy
\newmacro{\stratalt}{\alt\strat}		% for alternate strategy
\newmacro{\strataltalt}{\altalt\strat}		% for second alternate
\newmacro{\strats}{\mathcal{X}}		% for set of mixed strategies
\newmacro{\intstrats}{\strats^{\circle}}		% for set of interior strategies

\newcommand{\eq}{\sol}		% for Nash equilibrium
\newmacro{\loss}{\ell}		% for loss function
\newmacro{\pay}{u}		% for payoff function
\newmacro{\payv}{v}		% for payoff vector
\newmacro{\pot}{f}		% for potential function

\newmacro{\game}{\mathcal{G}}		% for game
\newmacro{\gameall}{\game(\players,\points,\loss)}		% for game with all elements

\newmacro{\fingame}{\Gamma}		% for finite game
\newmacro{\fingameall}{\Gamma(\players,\pures,\pay)}		% for finite game with all elements

\newmacro{\gmat}{g}		% for metric tensor
\newmacro{\gdist}{\dist_{\gmat}}
\newmacro{\mfld}{M}		% for manifold
\newmacro{\form}{\omega}		% for generic form

\newmacro{\tvec}{z}		% for tangent vector
\newmacro{\uvec}{u}		% for unit vector

\newmacro{\ball}{\mathbb{B}}		% for ball
\newmacro{\sphere}{\mathbb{S}}		% for sphere

\newmacro{\vertex}{v}		% for vertex
\newmacro{\vertexalt}{\alt\vertex}		% for alternate vertex
\newmacro{\vertexaltalt}{\altalt\vertex}		% for second alternate
\newmacro{\nVertices}{V}		% for number of vertices
\newmacro{\vertices}{\mathcal{\nVertices}}		% for set of vertices

\newmacro{\edge}{e}		% for edge
\newmacro{\edgealt}{\alt\edge}		% for alternate edge
\newmacro{\edgealtalt}{\altalt\edge}		% for second alternate
\newmacro{\nEdges}{E}		% for number of edges
\newmacro{\edges}{\mathcal{\nEdges}}		% for set of edges

\newmacro{\graph}{\mathcal{G}}		% for graph
\newmacro{\graphall}{\graph(\vertices,\edges)}		% for graph with all elements

\newmacro{\vecspace}{\mathcal{V}}		% for generic vector space
\newmacro{\subspace}{\mathcal{W}}		% for vector subspace

\newmacro{\bvec}{e}		% for basis vector
\newmacro{\bvecs}{\mathcal{E}}		% for basis vectors

\newmacro{\coord}{i}		% for coordinate index
\newmacro{\coordalt}{j}		% for alternate coordinate
\newmacro{\coordaltalt}{k}		% for second alternate
\newmacro{\nCoords}{d}		% for number of coordinates
\newmacro{\dims}{\nCoords}		% for dimension
\newmacro{\vdim}{\nCoords}		% for dimension (legacy alias)

\newmacro{\pspace}{\mathcal{X}}		% for primal space
\newmacro{\dspace}{\mathcal{Y}}		% for dual space

\newmacro{\ppoint}{x}		% for generic primal point
\newmacro{\ppointalt}{\alt\ppoint}		% for alternate primal point
\newmacro{\ppointaltalt}{\altalt\ppoint}		% for second alternate
\newmacro{\ppoints}{\mathcal{X}}		% for set of primal points
\newmacro{\pstate}{X}		% for primal state
\newmacro{\pvec}{u}		% for primal displacement vector

\newmacro{\dpoint}{y}		% for generic dual point
\newmacro{\dpointalt}{\alt\dpoint}		% for alternate dual point
\newmacro{\dpointaltalt}{\altalt\dpoint}		% for second alternate
\newmacro{\dpoints}{\mathcal{Y}}		% for set of dual points
\newmacro{\dstate}{Y}		% for dual state
\newmacro{\dvec}{v}		% for dual displacement vector

\newmacro{\mat}{M}		% for generic matrix
\newmacro{\hmat}{H}		% for Hessian matrix

\newmacro{\ones}{\mathbf{1}}		% for matrix of ones
\newmacro{\eye}{I}		% for identity matrix
\newmacro{\zer}{\mathbf{0}}		% for zero matrix

\DeclarePairedDelimiter{\norm}{\lVert}{\rVert}		% for norm
\DeclarePairedDelimiterXPP{\dnorm}[1]{}{\lVert}{\rVert}{_{\ast}}{#1}		% for dual norm

\DeclarePairedDelimiterXPP{\onenorm}[1]{}{\lVert}{\rVert}{_{1}}{#1}		% for L1 norm
\DeclarePairedDelimiterXPP{\twonorm}[1]{}{\lVert}{\rVert}{_{2}}{#1}		% for L2 norm
\DeclarePairedDelimiterXPP{\supnorm}[1]{}{\lVert}{\rVert}{_{\infty}}{#1}		% for sup norm

\DeclarePairedDelimiterX{\braket}[2]{\langle}{\rangle}{#1\mathopen{}\delimsize\vert\mathopen{}#2}

\DeclarePairedDelimiterX{\inner}[2]{\langle}{\rangle}{#1,#2}		% for scalar product

\newcommand{\defeq}{\coloneqq}		% for direct definition
\newcommand{\from}{\colon}		% for function definition
\newmacro{\source}{O}		% for origin
\newmacro{\sink}{D}		% for destination

\newmacro{\pair}{i}		% for pair index
\newmacro{\pairalt}{j}		% for alternate pair
\newmacro{\pairaltalt}{k}		% for second alternate
\newmacro{\nPairs}{N}		% for number of pairs
\newmacro{\pairs}{\mathcal{\nPairs}}		% for set of pairs

\newmacro{\route}{p}		% for path
\newmacro{\routealt}{\alt\route}		% for alternate path
\newmacro{\routealtalt}{\altalt\route}		% for second alternate
\newmacro{\nRoutes}{P}		% for number of paths
\newmacro{\routes}{\mathcal{\nRoutes}}		% for set of paths

\newmacro{\flow}{f}		% for flow profile
\newmacro{\flowalt}{\alt\flow}		% for alternate flow
\newmacro{\flowaltalt}{\altalt\flow}		% for second alternate
\newmacro{\flows}{\mathcal{F}}		% for set of flows

\newmacro{\load}{x}		% for load profile
\newmacro{\loadalt}{\alt\load}		% for alternate load
\newmacro{\loadaltalt}{\altalt\load}		% for second alternate
\newmacro{\loads}{\mathcal{X}}		% for set of loads

\newop{\Opt}{Opt}		% for value of problem
\newop{\Sol}{Sol}		% for solution of problem
\newop{\gap}{Gap}		% for gap function
\newop{\orcl}{Or}		% for oracle

\newmacro{\obj}{f}		% for objective function
\newmacro{\objalt}{g}		% for alternate objective (smooth etc.)
\newmacro{\sobj}{F}		% for stochastic objective

\newmacro{\gvec}{g}		% for gradient vector
\newmacro{\oper}{A}		% for operator
\newmacro{\vecfield}{v}		% for vector field (selection etc.)

\newcommand{\sol}[1][\point]{#1^{\ast}}		% for solution point (x by default)
\newmacro{\vbound}{G}		% for vector bound
\newmacro{\lips}{L}		% for Lipschitz modulus
\newmacro{\strong}{\alpha}		% for strong convexity modulus
\newmacro{\smooth}{\beta}		% for strong smoothness modulus

\newop{\tspace}{T}		% for tangent space
\newop{\tcone}{TC}		% for tangent cone
\newop{\dcone}{\tcone^{\ast}}		% for dual cone
\newop{\ncone}{NC}		% for normal cone
\newop{\pcone}{PC}		% for polar cone
\newop{\hull}{\Delta}		% for simplices

\newmacro{\cvx}{\mathcal{C}}		% for generic convex set
\newmacro{\subd}{\partial}		% for subdifferential

\newmacro{\minmax}{L}		% for minmax objective

\newmacro{\minvar}{\theta}		% for minimization variable
\newmacro{\minvaralt}{\alt\minvar}		% for alternate minvar
\newmacro{\minvars}{\Theta}		% for set of minvars

\newmacro{\maxvar}{\phi}		% for maximization variable
\newmacro{\maxvaralt}{\alt\maxvar}		% for alternate maxvar
\newmacro{\maxvars}{\Phi}		% for set of maxvars

\newmacro{\hreg}{h}		% for regularizer
\newmacro{\breg}{D}		% for Bregman divergence
\newmacro{\mprox}{P}		% for mirror prox-mapping
\newmacro{\mirror}{Q}		% for mirror map
\newmacro{\fench}{F}		% for Fenchel coupling
\newmacro{\hstr}{K}		% for strong convexity constant
\newmacro{\depth}{H}		% for regularizer depth
\newmacro{\proxdom}{\points^{\hreg}}		% for prox-domain

\DeclarePairedDelimiterXPP{\proxof}[2]{\mprox_{#1}}{(}{)}{}{#2}		% for Bregman prox step

\newmacro{\zone}{\mathbb{D}}		% for Bregman zone
\newop{\Eucl}{\Pi}		% for Euclidean projection
\newop{\logit}{\Lambda}		% for logit map
\newop{\dkl}{KL}		% for Kullback Leibler

\newmacro{\point}{x}		% for generic point
\newmacro{\pointalt}{\alt\point}		% for alternate point
\newmacro{\pointaltalt}{\altalt\point}		% for second alternate
\newmacro{\points}{\mathcal{K}}		% for set of points
\newmacro{\intpoints}{\relint\points}		%for point set interior

\newmacro{\base}{p}		% for reference point
\newmacro{\basealt}{q}		% for alternate reference point
\newmacro{\basealtalt}{u}		% for second alternate

\newmacro{\open}{\mathcal{U}}		% for open sets
\newmacro{\closed}{\mathcal{C}}		% for closed sets
\newmacro{\cpt}{\mathcal{K}}		% for compact sets
\newmacro{\nhd}{U}		% for neighborhoods

\newop{\ex}{\mathbb{E}}		% for expectations
\newop{\prob}{\mathbb{P}}		% for probability
\newop{\Var}{Var}		% for variance
\newop{\simplex}{\hull}		% for simplices

\providecommand\given{}		% empty command for conditionals

\DeclarePairedDelimiterXPP{\exof}[1]{\ex}{[}{]}{}{%		% for conditional expectations
\renewcommand\given{\nonscript\,\delimsize\vert\nonscript\,\mathopen{}} #1}

\DeclarePairedDelimiterXPP{\probof}[1]{\prob}{(}{)}{}{%		% for conditional probabilities
\renewcommand\given{\nonscript\:\delimsize\vert\nonscript\:\mathopen{}} #1}

\DeclarePairedDelimiterXPP{\oneof}[1]{\one}{\{}{\}}{}{%		% for conditional expectations
\renewcommand\given{\nonscript\,\delimsize\vert\nonscript\,\mathopen{}} #1}

\newmacro{\sample}{\omega}		% for sample
\newmacro{\samples}{\Omega}		% for set of samples

\newmacro{\filter}{\mathcal{F}}		% for filtration
\newmacro{\probspace}{(\samples,\filter,\prob)}		% for probability space

\newmacro{\event}{E}       % for event
\newmacro{\eventalt}{H}       % for alternate event
\newmacro{\mean}{\mu}		% for mean of distribution
\newmacro{\sdev}{\sigma}		% for mean of distribution
\newmacro{\variance}{\sdev^{2}}		% for mean of distribution

\newmacro{\step}{\gamma}		% for step-size
\newmacro{\learn}{\eta}		% for learning rate

\newmacro{\proper}{\tau}		% for proper time
\newmacro{\signal}{V}		% for signal
\newmacro{\error}{\xi}		% for error
\newmacro{\noise}{Z}		% for noise
\newmacro{\bias}{b}		% for bias
\newmacro{\brown}{W}		% for Wiener process

\newmacro{\serror}{\theta}		% for scalar error
\newmacro{\snoise}{\xi}		% for scalar noise
\newmacro{\sbias}{\psi}		% for scalar bias

\newmacro{\sbound}{M}		% for signal bound
\newmacro{\bbound}{B}		% for bias bound
\newmacro{\noisepar}{\sdev}		% for noise parameter
\newmacro{\noisevar}{\variance}		% for noise variance

\newcommand{\parensnew}[1]{\left(#1\right)}
\newcommand{\absnew}[1]{\left|#1\right|}
\usepackage{comment}
\usepackage[skins]{tcolorbox}
\usepackage{pgfplots}
\pgfplotsset{width=7cm,compat=1.8}
\usepackage{pgfplotstable}
\usepackage{tikz-3dplot}
\usepackage{enumitem}

\newcommand{\currplay}[1][\state]{\debug{#1}_{\play,\run}}		% for current iterate (X by default)
\newcommand{\nextplay}[1][\state]{\debug{#1}_{\play,\run+1}}		% for next iterate (X by default)

\newmacro{\estate}{\hat{X}} %for exploration state
\newmacro{\npay}{\hat{u}} %for noisy payoff function
\newmacro{\npayv}{\hat{v}} %for noisy payoff vector
\newmacro{\epar}{\eps} %for exploration parameter
\newmacro{\round}{n}
\newmacro{\subround}{k}
\newmacro{\Lipscon}{C}
\newmacro{\stcon}{K}
\newmacro{\paybound}{{\beta}'}
\newmacro{\noiselb}{\beta} %noise lower bound
\newmacro{\problb}{\pi} %probability of noise's lower bound
\newmacro{\conlevel}{\delta}
\newmacro{\PCsym}{G} %symbol for elements in the polar cone
\newmacro{\varbound}{\sigma^2}
\newmacro{\semibandit}{Limited first-order feedback}

\newmacro{\totalbound}{\mu}
\newmacro{\antivarpi}{\varrho }

\addauthor[\textbf{Angel}]{AG}{DarkOrange2}

\addauthor[\textbf{Manolis}]{MV}{Maroon0}

\makeatletter
\def\namedlabel#1#2{\begingroup
    #2%
    \def\@currentlabel{#2}%
    \phantomsection\label{#1}\endgroup
}
\makeatother

\makeatletter
\newcommand{\manuallabel}[2]{\def\@currentlabel{#2}\label{#1}}
\makeatother

\usepackage{fancybox}
\usepackage{wrapfig}
\newmacro{\rvdrift}{\mathsf{drift}}
\newmacro{\rvnoise}{\mathsf{noise}}
\newmacro{\rvbias}{\mathsf{bias}}

\addauthor[\textbf{Pan}]{PM}{blue}

\newacro{LHS}{left-hand side}
\newacro{RHS}{right-hand side}
\newacro{iid}[i.i.d.]{independent and identically distributed}
\newacro{lsc}[l.s.c.]{lower semi-continuous}
\newacro{NE}{Nash equilibrium}
\newacroplural{NE}[NE]{Nash equilibria}
\newacro{CCE}{coarse correlated equilibrium}
\newacroplural{CCE}[CCE]{coarse correlated equilibria}

\newacro{MAB}{multi-armed bandit}
\newacro{FTRL}{``follow the regularized leader''}
\newacro{EW}{exponential weights}
\newacro{IWE}{importance-weighted estimator}

\title
[Stable and unstable equilibria under regularized learning]
{Survival of the strictest:
Stable and unstable equilibria under regularized learning with partial information}
\coltauthor{%
 \Name{Angeliki Giannou} \Email{giannouangeliki@gmail.com}\\
 \addr National Technical University of Athens\\
 \Name{Emmanouil-Vasileios Vlatakis-Gkaragkounis} \Email{emvlatakis@cs.columbia.com}\\
 \addr Columbia University\\
 \Name{Panayotis Mertikopoulos} \Email{panayotis.mertikopoulos@imag.fr}\\
 \addr Univ. Grenoble Alpes, CNRS, Inria \& Criteo AI Lab%
}

\begin{document}

%----------------------------------------------------------------------
%%% ABSTRACT
%----------------------------------------------------------------------
\maketitle
\begin{abstract}
%----------------------------------------------------------------------
%%% ABSTRACT
%----------------------------------------------------------------------
% !TEX root = ./Main.tex

In this paper, we examine the \acl{NE} convergence properties of no-regret learning in general $\nPlayers$-player games.
For concreteness, we focus on the archetypal \acdef{FTRL} family of algorithms, and we consider the full spectrum of uncertainty that the players may encounter \textendash\ from noisy, oracle-based feedback, to bandit, payoff-based information.
In this general context, we establish a comprehensive equivalence between the stability of a \acl{NE} and its support:
\emph{a \acl{NE} is stable and attracting with arbitrarily high probability if and only if it is strict} (\ie each equilibrium strategy has a unique best response).
This equivalence extends existing continuous-time versions of the ``folk theorem'' of evolutionary game theory to a bona fide algorithmic learning setting,
and it provides a clear refinement criterion for the prediction of the day-to-day behavior of no-regret learning in games.

\end{abstract}

%----------------------------------------------------------------------
%%% KEYWORDS
%----------------------------------------------------------------------
%\begin{keywords}%
%No-regret learning,
%\emph{Follow The Regularized Leader}-(FTRL), Asymptotic (In)-stability, (Mixed/Strict/Quasi) Nash Equilibrium. 
%\end{keywords}

%**********************************************************************
%***    BODY TEXT
%**********************************************************************
%\allowdisplaybreaks		% for breaking long displays
\acresetall		% for resetting acros

%----------------------------------------------------------------------
%%% TABLE OF CONTENTS
%----------------------------------------------------------------------
%\setcounter{tocdepth}{1}
%\tableofcontents

%----------------------------------------------------------------------
%%% INTRODUCTION
%----------------------------------------------------------------------
\section{Introduction}
\label{sec:introduction}

The prototypical framework for online learning in games can be summarized as follows:
\begin{enumerate}
[itemsep=0pt]
\item
At each stage of the process, every participating agent chooses an action from some finite set.
\item
All agents receive a reward based on the actions of all other players and their individual payoff functions (assumed a priori unknown).
\item
The players record their rewards and any other feedback generated during the payoff phase, and the process repeats.
\end{enumerate}
This multi-agent framework has both important similarities and major differences with \emph{single-agent} online learning.
Indeed, if we isolate a single, focal player and abstract away all others,
we essentially recover a \ac{MAB} problem \textendash\ stochastic or adversarial, depending on the assumptions for the non-focal players \citep{CBL06,BCB12}.
In this case, the most widely used figure of merit is the agent's \emph{regret}, \ie the difference between the agent's cumulative payoff and that of the best fixed action in hindsight.
Accordingly, much of the literature on online learning has focused on deriving regret bounds that are min-max optimal, both in terms of the horizon $\nRuns$ of the process, as well as the number of actions $\nPures$ available to the focal player.

On the other hand, from a game-theoretic standpoint, the main question that arises is whether players eventually settle on an equilibrium profile from which no player has an incentive to deviate.
In this regard, a ``folk'' result states that the empirical frequency of play under no-regret play converges to the game's set of \acdefp{CCE} \citep{Han57,HMC00}.
However, there are two key caveats with this result.
First, \acp{CCE} are considerably weaker than \aclp{NE}, to the extent that they fail even the most basic postulates of rationalizability \citep{DF90}:
as was shown by \citet{VZ13}, \acp{CCE} may be supported \emph{exclusively} on \emph{strictly dominated} strategies, even in simple, symmetric two-player games.
Second, the convergence of the empirical mean does not carry any tangible guarantees for the players' day-to-day behavior:
under this type of convergence, the player's best payoff over time could be close to that of a \acl{NE}, but the players might otherwise be spending arbitrarily long periods of time on dominated strategies.

The above is just a well-known example of the convergence failures of no-regret learning in games with a possibly exotic equilibrium structure.
More to the point, even when the underlying game admits a \emph{unique} \acl{NE}, recent works have shown that no-regret algorithms \textendash\ such as the popular multiplicative weights update (MWU) method \textendash\ could still lead to chaotic \citep{PPP17,CP19,MP17} or Poincaré recurrent / cycling behavior \citep{HSV09,MS16,MPP18}.
From a convergence viewpoint, all these results can be seen as instances of a much more general impossibility result at play:
\emph{there are no uncoupled dynamics leading to \acl{NE} in \emph{all} games} \citep[\citeauthor{HMC03},][]{HMC03}.%
\footnote{``Uncoupled'' means here that each player's update rule does not depend explicitly on the payoffs of other players.}
Since no-regret dynamics are by definition unilateral, they are \emph{a fortiori} uncoupled, so this result shatters any hope of obtaining a universal \acl{NE} convergence result for the players' day-to-day behavior.

\para{Our contributions}

In view of the above, a critical question that arises is the following:
\emph{Is there a class of \aclp{NE} that consistently attract no-regret processes?}
Conversely,
\emph{are all \aclp{NE} equally likely to emerge as outcomes of a no-regret learning process?}

To address these questions in as general a setting as possible, we focus on the \acdef{FTRL} family of algorithms:
this is arguably the most widely used class of dynamics for no-regret learning in games, and it includes as special cases the seminal multiplicative weights\,/\,EXP3 algorithms \citep{SS11,SSS06,AHK12}.
In terms of feedback, we also consider a flexible, context-agnostic template in which players are only assumed to have access to an inexact model of their payoff vectors at a given stage.
This model for the players' feedback covers a broad range of modeling assumptions,
such as
\begin{enumerate*}
[(\itshape a\upshape)]
\item
the case where players can retroactively compute \textendash\ or otherwise observe \textendash\ their full payoff vectors (\eg as in routing games);
and
\item
the \emph{bandit} case, where players only observe their in-game payoffs and have no other information on the game being played.
\end{enumerate*}

The range of modeling assumptions covered by our framework is quite extensive, so one would likewise expect different, context-specific answers to these questions \textendash\ presumably with equilibria becoming ``less stable'' as information becomes ``more scarce''.
This expectation is justified by the behavior of no-regret learning in single-agent environments:
there, the type of information available to the learner has a dramatic effect on the achieved regret minimization rate.
Nevertheless, we show that this conjecture is \emph{false:}
as far as the algorithms' equilibrium convergence properties are concerned, the learning dynamics described above are all \emph{equivalent}.
\smallskip

In more detail, we show that all \ac{FTRL} algorithms under study enjoy the following properties:
\begin{enumerate}
[leftmargin=\parindent,label={\itshape \alph*\upshape)}]
\item
\emph{Strict \aclp{NE} are \emph{stochastically} asymptotically stable} \textendash\ \ie they are stable and attracting with arbitrarily high probability.
\item
\emph{Only strict \aclp{NE} have this property:}
mixed \aclp{NE} supported on more than one strategies are inherently unstable from a learning viewpoint.
\end{enumerate}
\smallskip
\noindent
We are not aware of a similar result in the literature at this level of generality (\ie including models with bandit feedback),
and we believe that this equivalence represents an important refinement criterion for the prediction of the day-to-day behavior of no-regret learners in the face of uncertainty and lack of perfect information.

\para{Related work}

To put our contributions in the proper context, we provide below an account of relevant works in the literature, classified along the two directions of our main result: ``strictness$\implies$stability'' and ``stability$\implies$strictness''.

\begin{enumerate}
[leftmargin=\parindent,label={\bfseries\Roman*.},itemsep=0pt]

\item
\textbf{Strictness$\implies$Stability.}
Analyzing the convergence of game-theoretic learning dynamics has generated a vast corpus of literature that is impossible to survey here.
Nonetheless, an emerging theme in this literature is the focus on specific classes of games (such as potential games or $2^{\nPlayers}$ games).
As a purely indicative \textendash\ and highly incomplete \textendash\ list, we cite here the works of \citef{LC05} and \citef{Les06}, \citef{CMS10}, \citef{KPT11}, \citef{CGM15}, \citef{SALS15}, and \citef{CHM17-NIPS}, who provide a range of equilibrium convergence results in potential, $2^{\nPlayers}$, and $(\lambda,\mu)$-smooth games, under different feedback assumptions \textendash\ from payoff vector observations \citep{KPT11,SALS15} to bandit \citep{LC05,Les06,CMS10,CHM17-NIPS}.
By contrast, our focus is determining the stochastic stability of a class of \emph{equilibria} \textendash\ not \emph{games}. %

As far as we are aware, the only comparable results in this literature concern an idealized continuous-time, deterministic, full-information version of our setting, which is common in applications to population biology and evolutionary game theory.
In this context, building on earlier results on the replicator dynamics \cite{Wei95,HS98}, the authors of \cite{MS16} showed that strict \aclp{NE} are asymptotically stable under the continuous-time dynamics of \ac{FTRL}.
However, we stress here again that these results only concern continuous-time, deterministic dynamical systems with an inherent full-information assumption;
we are not aware of a result providing convergence to strict \aclp{NE} with bandit feedback.

\item
\textbf{Stability$\implies$Strictness.}
In the converse direction, a related result in the literature on evolutionary games is that only strict \aclp{NE} are asymptotically stable under the (multi-population) replicator dynamics \cite{Wei95,San10,HS03},
a continuous-time, deterministic dynamical system which can be seen as the ``mean-field'' limit of the exponential weights algorithm \cite{Rus99,Sor09,HSV09}.
In a much more recent paper \cite{FVGL+20}, this implication was extended to the dynamics of \ac{FTRL}, but always in a deterministic, full-information, continuous-time setting.
In this regard, our results are aligned with \cite{FVGL+20};
however, other than this high-level conceptual link, there is no precise connection, either at the level of implications or at the level of proofs.
Specifically, the analysis of \cite{FVGL+20} relies crucially on volume-conservation arguments that are neither applicable nor relevant in a discrete-time stochastic setting \textendash\ where the various processes involved could jump around stochastically without any regard for volume contraction or expansion.
\end{enumerate}

\vspace{-0.25cm}
\para{Proof techniques}Learning with partial information is an inherently stochastic process, so our results are also stochastic in nature \textendash\ hence the requirement for asymptotic stability with arbitrarily high probability.
This constitutes a major point of departure from continuous-time models of learning \cite{MS16,FVGL+20}, so
our proof techniques are also radically different as a result.
The principal challenge in our proof of stability of strict \aclp{NE} comes in controlling the aggregation of error terms with possibly unbounded variance (coming from inverse propensity scoring of bandit-type observations).
Because of this, stochastic approximation techniques that have been used to show convergence with $L^{2}$-bounded feedback \cite{MZ19} cannot be applied in this setting;
we achieve this control by applying a sharp version of the Doob-Kolmogorov maximal inequality to control equilibrium deviations with high probability.
In the converse direction, the crucial argument in the proof of the \emph{instability} of mixed equilibria is a direct probabilistic estimate which leverages a non-degeneracy argument for the noise entering the process;
we are not aware of other works using a similar technique.
%
%
%

%----------------------------------------------------------------------
%%% PRELIMS
%----------------------------------------------------------------------
\section{Preliminaries}
\label{sec:prelims}
\para{The stage game}
Throughout this work we will focus on normal form games with a finite number of players and a finite number of actions per player.
Formally, such a game is defined as a tuple $\fingame = \fingameall$ with the following primitives:
\begin{itemize}
[leftmargin=\parindent,itemsep=0pt]
\item
A finite set of \emph{players} \textendash\ or \emph{agents} \textendash\ indexed by $\play \in \players = \braces{1,\dotsc,\nPlayers}$.
\item
A finite set of \emph{actions} \textendash\ or \emph{pure strategies} \textendash\ indexed by $\pure_{\play} \in \pures_{\play} = \braces{1,\dotsc,\nPures_{\play}}$, $\play\in\players$.
Players can also play \emph{mixed strategies}, which represent probability distributions $\strat_{\play} \in \strats_{\play} \defeq \simplex(\pures_{\play})$;
in this case, we will write $\strat_{\play\pure_{\play}}$ for the probability that player $\play\in\players$ selects $\pure_{\play}\in\pures_{\play}$.
Aggregating over all players, we will also write $\strat = (\strat_{1},\dotsc,\strat_{\nPlayers})$ for the players' \emph{mixed strategy profile} and $\strats \defeq \prod_{\play} \strats_{\play}$ for the set thereof.
Finally, when we want to focus on the strategy (or action) of a particular player $\play\in\players$, we will use the shorthand $(\strat_{\play};\strat_{-\play}) \defeq (\strat_{1},\dotsc,\strat_{\play},\dotsc,\strat_{\nPlayers})$ \textendash\ and, similarly, $(\pure_{\play};\pure_{-\play})$ for pure strategies.
\item
An ensemble of \emph{payoff functions} $\pay_{\play} \from \pures \to \R$ where $\pures \defeq \prod_{\play} \pures_{\play}$ is the space of all pure strategy profiles.
The expected payoff of player $\play$ in a mixed strategy profile $\strat\in\strats$ is then given by
\begin{equation}
\label{eq:pay}
\pay_{\play}\parens{\strat}
	\equiv \pay_{\play}\parens{\strat_{\play};\strat_{-\play}}
	= \sum_{\pure_{1}\in\pures_{1}} \dotsi \sum_{\pure_{\nPlayers}\in\pures_{\nPlayers}}
		\pay_{\play}\parens{\pure_{1},\dotsc,\pure_{\nPlayers}}
			\cdot \strat_{1,\pure_{1}} \dotsm \strat_{\nPlayers,\pure_{\nPlayers}}
\end{equation}
where $\pay_{\play}\parens{\pure_{1},\dotsc,\pure_{\nPlayers}}$ is the payoff of player $\play$ in the action profile $\pure =\parens{\pure_{1},\dotsc,\pure_{\nPlayers}}\in \pures$.
\end{itemize}

For posterity, we will also write
\(
\payv_{\play\pure_{\play}}\parens{\strat}
	= \pay_{\play} \parens{\pure_{\play};\strat_{-\play}}
\)
for the payoff that player $\play$ would have gotten by playing $\pure_{\play}\in\pures_{\play}$ against the mixed strategy profile $\strat_{-\play}$ of all other players.
In this way, the \emph{mixed payoff vector} of the $\play$-th player will be
\begin{equation}
\label{eq:indpay-mixed}
\payv_{\play}(\point)
	= (\payv_{\play\pure_{\play}}(\point))_{\pure_{\play}\in\pures_{\play}}
\end{equation}
and we will write $\payv(\point) = (\payv_{1}(\point),\dotsc,\payv_{\nPlayers}(\point))$ for the ensemble thereof.
For notational convenience, we will also set $\dspace_{\play} = \R^{\pures_{\play}}$ and $\dspace = \prod_{\play}\dspace_{\play}$ for the space of payoff vectors and profiles respectively.
Finally, in a slight abuse of notation, we will identify $\pure_{\play}$ with the mixed strategy that assigns all probability to $\pure_{\play}$, and we will denote the corresponding \emph{pure payoff vector} as
\(
\payv_{\play}\parens{\pure}
	= \parens{\pay_{\play}\parens{\pure_{\play};\pure_{-\play}}}_{\pure_{\play}\in\pures_{\play}}.
\)
The distinction between pure and mixed payoff vectors will become important later on, when we discuss the information at each player's disposal.

\para{\acl{NE}}

In this general context, the most widely used solution concept is that of a \acl{NE}, \ie a mixed strategy profile that discourages unilateral deviations.
Formally, $\eq$ is a \acli{NE} of $\fingame$ if
\begin{equation}
\label{eq:NEgen}
\tag{NE}
\pay_{\play}(\eq)
	\geq \pay_{\play}(\strat_{\play};\eq_{-\play})
	\quad
	\text{for all $\strat_{\play}\in\strats_{\play}$ and all $\play\in\players$}.
\end{equation}
The set of pure strategies supported at the equilibrium component $\eq_{\play} \in \strats_{\play}$ of each player will be denoted by $\supp\parens{\eq_{\play}} = \setdef{\pure_{\play}\in\pures_{\play}}{\eq_{\play\pure_{\play}} > 0}$.
Accordingly, \aclp{NE} can be equivalently characterized by means of the variational inequality
\begin{equation}
\label{eq:NE}
\payv_{\play\eq[\pure_{\play}]}(\eq)
	\geq \payv_{\play\pure_{\play}}(\eq)
	\quad
	\text{for all $\eq[\pure_{\play}]\in\supp\parens{\eq_{\play}}$ and all $\pure_{\play} \in \pures_{\play}$, $\play\in\players$}.
\end{equation}

\noindent
The above characterization gives rise to the following classification of \aclp{NE}:
\begin{itemize}
\setlength{\itemsep}{0pt}
\item
$\eq$ is a \emph{pure equilibrium} if $\supp\parens{\eq_{\play}}$ only contains a single strategy for all $\play\in\players$.
\item
$\eq$ is a \emph{mixed equilibrium} in any other case;
in particular, if $\supp\parens{\eq_{\play}} = \pures_{\play}$ for all $\play\in\players$, we say that $\eq$ is \emph{fully mixed}.
\end{itemize}

By definition,
pure equilibria correspond to vertices of $\strats$,
fully mixed equilibria lie in the relative interior $\relint(\strats)$ of $\strats$,
and, more generally,
mixed equilibria lie in the relative interior of the face of the simplex spanned by the support of each player's equilibrium component.
A further distinction between \aclp{NE} that is inherited by the inequality \eqref{eq:NE} is as follows:
if \eqref{eq:NE} holds as a strict inequality for all $\pure_{\play}\in\pures_{\play}\setminus \supp\parens{\eq_{\play}}$, $\play \in\players$, the equilibrium in question is said to be \emph{quasi-strict} \cite{FT91}.
Quasi-strict equilibria have the defining property that all pure best responses are played with positive probability;
it is also well known that all \aclp{NE} in all but a measure-zero set of games are quasi-strict.
For this reason, the property of having a quasi-strict equilibrium is generic, and games that enjoy this property are called themselves \emph{generic} [
Specifically, the set of games with \aclp{NE} that are not quasi-strict is \emph{meager} in the Baire category sense.]

We stress here that quasi-strict equilibria could be either mixed or pure:
for example, the equilibrium of Matching Pennies is fully mixed and quasi-strict, whereas the equilibrium of the Prisoner's dilemma is quasi-strict and pure.
In this last case, when a quasi-strict equilibrium is pure, it will be called \emph{strict}:
any deviation from an equilibrium strategy results in a strictly worse payoff.

%----------------------------------------------------------------------
%%% FTRL
%----------------------------------------------------------------------
\section{Regret minimization and regularized learning}
\label{sec:FTRL}
%----------------------------------------------------------------------
%%% FTRL
%----------------------------------------------------------------------
% !TEX root = ./Main.tex

A key requirement in the context of online learning is the minimization of the players' regret, \ie the cumulative payoff difference between each player's chosen action and the best possible action in hindsight over a given horizon of play $\nRuns$.
Formally, given a sequence of play $\state_{\run} \in \pspace$, $\run=\running$, the (external) \emph{regret} of player $\play\in\players$ is defined as
\begin{equation}
\label{eq:regret}
\reg_{\play}(\nRuns)
	= \max_{\strat_{\play}\in\strats_{\play}}
		\sum_{\run=\start}^{\nRuns} \bracks{
			\pay_{\play}(\strat_{\play};\state_{-\play,\run}) - \pay_{\play}(\state_{\play,\run};\state_{-\play,\run})
			} % end bracks
\end{equation}
and we will say that player $\play$ has no regret if $\reg_{\play}\parens{\nRuns} = o\parens{\nRuns}$. 

One of the most widely used online learning schemes to achieve this requirement is the so-called \acdef{FTRL} family of algorithms \cite{SSS06,SS11}.
Heuristically, at each stage of the learning process, \ac{FTRL} prescribes a mixed strategy that maximizes the player's (perceived) cumulative payoff modulo a regularization penalty whose role is to ``smooth out'' the transition between strategies during play.
%More precisely, the round-by-round sequence of events can be described as follows:
%\begin{enumerate}
%\setlength{\itemsep}{0pt}
%\setlength{\parskip}{.1ex}
%\item
%The players receive \textendash\ or otherwise construct \textendash\ an estimate of their mixed payoff vectors.
%\item
%This estimate is integrated in an auxiliary ``score vector'' that determines each player's propensity towards a given action.
%\item
%The players apply a ``regularized'' choice map to transform these propensities into mixed strategies, and the process repeats.
%\end{enumerate}
Formally, this leads to the round-by-round recursive rule
\begin{equation}
\label{eq:FTRL}
\tag{FTRL}
\begin{aligned}
\state_{\play,\run}
	&= \mirror_{\play}(\dstate_{\play,\run})
	\\
\dstate_{\play,\run+1}
	&= \dstate_{\play,\run}
		+ \step_{\run} \npayv_{\play,\run}
\end{aligned}
\end{equation}
where
$\mirror_{\play}\from\dspace_{\play} \to \strats_{\play}$ denotes the ``choice map'' of player $\play\in\players$,
$\step_{\run} > 0$ is a ``learning rate'' parameter such that $\sum_{\run}\step_{\run} = \infty$,
and
$\npayv_{\play,\run}$ is a ``payoff signal'' that provides an estimate for the mixed payoffs of player $\play$ at stage $\run$.
We discuss each of these components in detail below.

\subsection{The feedback model}
\label{sec:feedback}

Depending on the specific framework at play, the modeling details concerning the feedback received by the players may vary wildly.
For example, when modeling congestion in a city, it is reasonable to assume that commuters can estimate the time it would have taken them to get to their destination via a different route \textendash\ \eg by means of a GPS service or an app like GoogleMaps or Waze.
By contrast, in applications of online learning to auctions and online advertising, it is not clear how a player could estimate the payoff of actions they did not play.

To account for as broad a range of feedback models as possible, we will take a context-agnostic approach and assume that each player receives a ``black-box'' model of their payoff vector of the form
\begin{equation}
\label{eq:feedback}
%\tag{feedback model}
\npayv_{\run}
	= \payv\parens{\pstate_{\run}}
	+ \error_{\run}
%	+ \noise_{\run} + \bias_{\run}
\end{equation}
for some abstract error process $\error_{\run} = (\error_{\play,\run})_{\play\in\players}$.
To differentiate between random (zero-mean) and systematic (non-zero-mean) errors, we will further decompose $\error_{\run}$ as
%\begin{equation}
\(
\error_{\run}
	= \noise_{\run} + \bias_{\run},
\)
%\end{equation}
where
\begin{equation}
\bias_{\run}
	= \exof{\error_{\run} \given \filter_{\run}}
	\quad
	\text{and}
	\quad
\exof{\noise_{\run} \given \filter_{\run}}
	= 0
\end{equation}
with $\filter_{\run}$ denoting the history of $\state_{\run}$ up to stage $\run$ (inclusive)
\footnote{Of course, since the feedback signal is generated only \emph{after} the player chooses a strategy, $\npayv_{\run}$ is not $\filter_{\run}$-measurable in general.}.
%relative to some underlying (complete) probability space $\probspace$.%
%\footnote{Since the feedback signal is generated only \emph{after} the player chooses a strategy, $\npayv_{\run}$ is not $\filter_{\run}$-measurable in general.}
%we will posit throughout that $\npayv_{\run}$ is $\filter_{\run+1}$-measurable, but not necessarily $\filter_{\run}$-measurable.}
We may then characterize the input signal $\curr[\npayv]$ by means of the following statistics:
\vspace{-\smallskipamount}
\begin{subequations}
\label{eq:sigbounds}
\begin{alignat}{3}
\label{eq:bbound}
a)
	\quad
	&\textit{Bias:}
	&\hspace{2em}
%	\bbound_{\run}
%		&= \exof{\dnorm{\bias{\run}} \given \filter_{\run}}
	&\exof{\dnorm{\bias_{\run}} \given \filter_{\run}}
		\leq \bbound_{\run}
	\\
%\label{eq:variance}
%b)
%	\quad
%	&\textit{Conditional variance:}
%	&\hspace{2em}
%%	\curr[\sdev]^{2}
%%		&= \exof{\dnorm{\noise{\run}}^{2} \given \filter_{\run}}
%	&\exof{\dnorm{\noise_{\run}}^{2} \given \filter_{\run}}
%		\leq \sdev_{\run}^{2}
%	\hspace{13em}
b)
	\quad
	&\textit{Mean square:}
	&\hspace{2em}
	&\exof{\dnorm{\npayv_{\run}}^{2} \given \filter_{\run}}
		\leq \sbound_{\run}^{2}
	\hspace{15em}
\end{alignat}
\end{subequations}
In the above, $\bbound_{\run}$ and $\sbound_{\run}$ represent deterministic bounds on the bias and variance of the feedback signal $\npayv_{\run}$.
%Depending on their value, a feedback model with $\bbound_{\run}=0$ will be called \emph{unbiased}, and a signal with $\lim_{\run\to\infty}\bbound_{\run} = 0$ will be called \emph{asymptotically unbiased};
%finally, an unbiased model with $\curr[\sdev] = 0$  will be called \emph{perfect}.
%\smallskip
For concreteness, we will also make the following blanket assumptions:
\begin{enumerate}
[left=1ex,label={\upshape({A}\arabic*)}]
\item
\label{asm:bias}
\emph{Bias control:}
$\lim_{\run\to\infty} \bbound_{\run} = 0$ and $\sum_{\run} \step_{\run} \bbound_{\run} < \infty$.
\item
\label{asm:variance}
\emph{Variance control:}
$\sum_{\run} \step_{\run}^{2} \sbound_{\run}^{2} < \infty$.
\item
\label{asm:degen}
\emph{Generic observation errors at equilibrium:}
For every mixed \acl{NE} $\eq$ of $\fingame$ and for all $\run=\running$, there exists a player $\play\in\players$ and strategies $a,b\in\supp(\eq_{\play})$ such that
%for all sufficiently small $\noiselb > 0$, there exists some $\problb>0$ with
\begin{equation}
\probof{\abs{\npayv_{\play a,\run}-\npayv_{\play b,\run}}\geq \noiselb \given \filter_{\run}}
	> 0
%	\geq \problb.
	\quad
	\text{for all sufficiently small $\noiselb > 0$}.
\end{equation}
%for all sufficiently small $\noiselb > 0$.
\end{enumerate}
%    Non-degeneracy: For each mixed Nash equilibrium $\eq$ of the game, there exist player $\play$, strategies $a,b\in \supp\parens{\eq_\play}$, such that at each round $\run\geq 1$ there exist $\problb_\run >0$, $\noiselb_\run>0$ and holds that
%    \[\probof{\abs{\npayv_{\play a,\run}-\npayv_{\play b,\run}}\geq \noiselb_\run} \geq \problb_\run\text{ for all }\run\geq 1 \]
%Below we provide two well-known models that satisfy all these assumptions.
The formulation of these hypotheses has been kept intentionally abstract because we have not made any modeling assumptions for how the players' payoff signals are generated.
In this regard, they are to be construed as an ``inexact model'' that allows for a wide variety of settings;
as an application, we illustrate below how these assumptions are verified in two widely used learning frameworks.

\begin{model}
[Oracle-based feedback]
%[Semi-bandit feedback]
\label{semibandit feedback}
Assume that each player chooses an action based on a given mixed strategy.
Then, once this procedure has been completed, an oracle reveals to each player the payoffs corresponding to their pure strategies given the other players' chosen strategies \textpar{in the congestion example, this oracle could be Waze or a GPS device}.
%and the pure strategies chosen by all others.
%\PM{I do not understand this sentence.
%Also, it does not seem consistent with the next one.}
Formally, at each round $\run$, every player $\play\in\players$ picks an action $\currplay[\pure]\in\pures_{\play}$ based on $\currplay\in\strats_{\play}$ and observes the pure payoff vector $\payv_{\play}\parens{\pure_{\run}}\equiv \parens{\pay_{\play}\parens{\pure_\play;\pure_{-\play,\run}}}_{\pure_\play\in\pures_{\play}}$.
Then the player's feedback signal is $\npayv_{\play,\run} = \payv_{\play}\parens{\pure_{\run}}$, which is a special case of the model \eqref{eq:feedback} with $\noise_{\run} = \payv\parens{\pstate_{\run}} - \payv\parens{\pure_{\run}}$ and $\bias_{\run} = 0$.
In more detail, we have:
\begin{itemize}[topsep=.5ex,leftmargin=\parindent]
\setlength{\itemsep}{0pt}
\setlength{\parskip}{.2ex}
\item
\ref{asm:bias} is trivial because $\bias_{\run} =0$.
%\ref{asm:bias} is trivial because $\exof{\npayv_{\run} \given \filter_{\run}} = \ex_{\state_{\run}}[\payv(\pure_{\run})] = \payv(\state_{\run})$, \ie $\bias_{\run} = 0$.
\item
\ref{asm:variance} is satisfied as long as $\sum_{\run}\step_{\run}^{2} <\infty$
\textpar{because $\sup_{\run}\exof{\dnorm{\npayv_{\run}}^{2}} \leq \max_{\strat} \dnorm{\payv(\strat)}^{2} < \infty$}.
\item
\ref{asm:degen} is proved in \cref{appendix instability}. 
\end{itemize}
\end{model}

\begin{model}
[Payoff-based feedback]
%[Bandit feedback]
\label{bandit feedback}
Assume that each player picks an action
%$\pure_{\play,\run}\in\pures_{\play}$
based on some mixed strategy as above;
however, players now only observe their realized payoffs $\pay_{\play}\parens{\pure_{\play,\run};\pure_{-\play,\run}}$.
This is the standard model for \aclp{MAB} \cite{CBL06,BCB12}, and it is also known as the ``bandit feedback'' setting.
In this case, players can estimate their payoff vectors by means of the \acl{IWE}:
%%\begin{equation}
%%\label{IWE}
%\tag{IWE} 
%%\npayv_{\play\pure_{\play},\run} = \left\{\begin{matrix}
%% \pay_{\play}\parens{\pure_{\play,\run};\pure_{-\play,\run}}/\estate_{\play\pure_{\play,\run}},&\pure_{\play} = \pure_{\play,\run} \\ 
%% 0,& \pure_{\play} \neq \pure_{\play,\run} 
%%\end{matrix}\right.
%%\end{equation}
\begin{equation}
\label{eq:IWE}
\tag{IWE} 
\npayv_{\play\pure_{\play},\run}
	= \frac{\oneof{\pure_{\play,\run} = \pure_{\play}}}{\estate_{\play\pure_{\play,\run}}}
		\pay_{\play}(\pure_{\run})
%	\frac{ \pay_{\play}\parens{\pure_{\play,\run};\pure_{-\play,\run}}}{\estate_{\play\pure_{\play,\run}}}\oneof{\pure_{\play} = \pure_{\play,\run}}
\end{equation}
where $\estate_{\play,\run} = \parens{1-\epar_{\run}}\pstate_{\play,\run} + \epar_{\run}/\abs{\pures_{\play}}$ is the mixed strategy of the $\play$-th player at stage $\run$.
Compared to $\state_{\play,\run}$, the player's actual sampling strategy is recalibrated by an explicit exploration parameter $\epar_{\run}\to0$ whose role is to stabilize the learning process by controlling the variance of \eqref{eq:IWE}.
%is a convex combination of the sequence of play $\pstate_{\play,\run}$, produced by \eqref{eq:FTRL} and the uniform distribution over the action set of each player $\pures_{\play}$. 
%It has not escaped our notice that running \eqref{eq:FTRL} with \eqref{eq:IWE} estimator directly with $\pstate_{\play,\run}$ could lead to uncontrollably large variance, due to potentially arbitrary small probability attributed to one or more strategies
%%, leading to instability of the model. 
%To stabilize \eqref{eq:FTRL} it is very common in the literature to add an explicit-exploration parameter, which is decreasing through time \ie $\epar_{\run}\to0$ and choose their strategies based on the recalibrated distribution $\currplay[\estate]$.
The idea is that even if a strategy has zero probability to be chosen under $\state_{\run}$, it will still be sampled with positive probability thanks to the mixing factor $\epar_{\run}$.
%the players will explore their other options by attributing positive probability to all of the strategies.\\

A standard calculation \textpar{that we defer to \cref{appendix instability}} shows that \eqref{eq:IWE} can be recast in the general form \eqref{eq:feedback} with $\bbound_{\run} = \bigoh(\epar_{\run})$ and $\sbound_{\run}^{2} = \bigoh(1/\epar_{\run})$.
We then have:
\begin{itemize}[topsep=.5ex,leftmargin=\parindent]
\setlength{\itemsep}{0pt}
\setlength{\parskip}{.2ex}
\item
\ref{asm:bias} is satisfied as long as $\epar_{\run}\to0$ and $\sum_{\run} \step_{\run}\epar_{\run} < \infty$.
\item
\ref{asm:variance} is satisfied as long as $\sum_{\run} \step_{\run}^{2} / \epar_{\run} < \infty$.
%\textpar{because $\sup_{\run}\exof{\dnorm{\npayv_{\run}}^{2} \given\filter_{\run}} < \infty$}.
\item
\ref{asm:degen} is proved in \cref{appendix instability}. 
\end{itemize}

%Again this model is entailed in our general \eqref{eq:feedback} with $\noise_{\run} =0$ and $\bias_{\run} = O\parens{\epar_{\run}}$, and all the assumptions \ref{asm:degen}-\ref{asm:variance} are satisfied:
%\begin{itemize}
%\item\ref{asm:degen} is proved in \cref{appendix A}.
%\item\eqref{A2} is trivial since $\noise_{\run} =0$.
%\item\ref{asm:bias} is satisfied if $\step_{\run}$, $\epar_{\run}$ are chosen appropriately such that $\sum_{\run} \step_{\run} \bbound_{\run} <\infty$.
%\item\ref{asm:variance} is satisfied if $\step_{\run},\epar_{\run}$ are chosen appropriately such that $\sum_{\run}\step_{\run}^{2} <\infty$, $\sum_{\run} \step_{\run}^{2}\noisepar_{\run}^{2} <\infty$, where $\noisepar_{\run}^{2} = \exof{\norm{\npayv_{\run}}^{2}_*\given \filter_{\run}} = \sum_{\play\in\players} \exof{\norm{\npayv_{\play,\run}}^{2}_*\given \filter_{\run}}\sim 1/\dis\min_{\pure_{\play}\in\pures_{\play}}\estate_{\play\pure_{\play}}\leq \abs{\pures_{\play}}/\epar_{\run}$.
%\end{itemize}
\end{model}

\begin{remark*}
The above conditions for the method's learning rate and exploration parameters can be achieved by using schedules of the form $\step_{\run} \propto 1/\run^{\pexp}$ and $\epar_{\run} \propto 1/\run^{\qexp}$ with $\pexp+\qexp>1$ and $2\pexp-\qexp > 1$.
A popular choice is $\pexp = 2/3 + \delta$ and $\qexp = 1/3 + \delta$ for some arbitrarily small $\delta>0$ \textendash\ or $\delta=0$ and including an extra logarithmic factor, \cf \cite{Sli19} and references therein.
\end{remark*}

%----------------------------------------------------------------------
%%% Regularization
%----------------------------------------------------------------------
\subsection{Regularization}
\label{sec:regularizers}

The second component of the \ac{FTRL} method is the players' ``choice map'' $\mirror_{\play}\from\dspace_{\play}\to\strats_{\play}$.
Because the players' score variables $\dstate_{\play,\run}$ essentially represent an estimate of each strategy's cumulative payoff over time, $\mirror_{\play}$ is defined as a ``regularized'' version of the best-response correspondence $\dpoint_{\play} \mapsto \argmax_{\strat_{\play}\in\strats_{\play}}\braces{\inner{\dpoint_{\play}}{\strat_{\play}}}$ (the regularization being necessary to avoid prematurely committing to a strategy).
On that account, we will consider \emph{regularized best responses} of the general form
\begin{equation}
\label{eq:choice}
\mirror_{\play}(\dpoint_{\play})
	= \argmax_{\strat_{\play}\in\strats_{\play}} \braces{ \inner{\dpoint_{\play}}{\strat_{\play}} - \hreg_{\play}(\strat_{\play}) }.
\end{equation}
In the above, each player's \emph{regularizer} $\hreg_{\play}\from\strats_{\play}\to\R$ is defined as $\hreg_{\play}(\strat_{\play}) = \sum_{\pure_{\play}\in\pures_{\play}} \theta_{\play}(\point_{\play})$ for some ``kernel function'' $\theta_{\play}\from[0,1]\to\R$ with the following properties:
\begin{enumerate*}[(\itshape i\hspace*{.5pt}\upshape)] 
\item
\manuallabel{h properties}{reguralizer's properties}
$\theta_{\play}$ is \emph{continuous} on $[0,1]$;
\item
$C^{2}$-smooth on $(0,1]$;
and
\item
$\inf_{[0,1]} \theta''_{\play} >0$.
\end{enumerate*}
%
%\smallskip
Of course, different regularizers give rise to different instances of \eqref{eq:FTRL};
for concreteness, we present below two prototypical examples thereof.

\begin{example}
[Multiplicative/Exponential weights update]
\label{MWU}
A popular choice of regularizer is the \textpar{negative} entropy $\hreg_i\parens{x} = \sum_{i}x_i\log x_i$, which leads to the \emph{logit choice} map $\Lambda_i\parens{y} = \exp\parens{y_i}/\sum_j \exp\parens{y_j}$ and the algorithm known as \emph{multiplicative weights update} \textpar{MWU}, \cf \cite{Vov90,LW94,ACBFS95,AHK12,SS11}.
\end{example}

\begin{example}[Euclidean projection]
\label{Projection}
Another popular regularizer is the quadratic penalty $\hreg_i\parens{x} = \sum_i{x_i}^{2}/2$, which yields the \emph{payoff projection} choice map $\Eucl_i\parens{y} = \argmin_{x\in \Delta}\norm{y-x}^{2}$, \cf \cite{Zin03,LS08}.
\end{example}

%----------------------------------------------------------------------
%%% RESULTS
%----------------------------------------------------------------------
\section{Analysis and Results}
\label{sec:results}

To understand the long-run behavior of \eqref{eq:FTRL}, we will focus on the following overarching question:
\emph{Which \aclp{NE} hold convergence and stability properties and how are these properties affected by the uncertainty in the players' feedback model?}

We provide the technical groundwork for our answers in \cref{sec:stability} below;
subsequently, we state our results in \cref{sec:statements}, and present the technical analysis in \cref{sec:Proof Techniques}.

\subsection{Asymptotic Stability}
\label{sec:stability} 

The first thing to note in this general context is that a game may admit several \aclp{NE}, both mixed and pure.
As a result, global convergence to an equilibrium from all initializations is not possible;
for this reason, we will focus on the notion of (\emph{stochastic}) \emph{asymptotic stability} \cite{HS98,San10,Kha12}.
Heuristically, an equilibrium is \emph{stochastically stable} if any sequence of play that begins close enough to the equilibrium in question, remains close enough with high probability;
in addition, if the sequence of play eventually converges to said equilibrium, then we say that it is stochastically asymptotically stable.
Formally, we have the following definition.

\begin{definition}
\label{stability}
Fix some arbitrary confidence level $\conlevel>0$.
Then $\eq\in\strats$ is said to be
\begin{enumerate}
[topsep=.5ex,label={\bfseries \arabic*.}]
\item
\textbf{Stochastically stable}
if, for every neighborhood $\nhd$ of $\eq$ in $\strats$, there exists a neighborhood $\nhd_{\start}$ of $\eq$ such that
whenever $\state_{\start} = \mirror(\dstate_{\start}) \in \nhd_{\start}$, we have
\begin{equation}
\probof{ \text{$\state_{\run}\in\nhd$ for all $\run=\running$} }
	\geq 1-\conlevel
\end{equation}
whenever $\state_{\start} = \mirror(\dstate_{\start}) \in \nhd_{\start}$.

\item
\textbf{Attracting}
if there exists a neighborhood $\nhd_{0}$ of $\eq$ such that
\begin{equation}
\probof{ \text{$\lim\nolimits_{\run\to\infty} \state_{\run} = \eq$}}
	\geq 1-\conlevel
\end{equation}
whenever $\state_{0} = \mirror(\dstate_{0}) \in \nhd_{\start}$.

\item
\textbf{Stochastically asymptotically stable}
if it is stochastically stable and attracting.
\end{enumerate}
\end{definition}

\noindent
\Cref{stability} will be the mainstay of our analysis and results, so some remarks are in order.

\begin{remark}
A first intricate detail in the above definition is the high probability requirement:
indeed, under uncertainty, a single unlucky estimation of the players' payoff vector could drive $\state_{\run}$ away from any neighborhood of $\eq$, possibly never to return.
In this regard, local stability results cannot be expected to hold with probability $1$, hence the requirement to hold with some arbitrary confidence level in the definition above.
\end{remark}

\begin{remark}
\label{stability remark}
Another remark worth making is the requirement $\state_{\start} = \mirror(\dstate_{\start}) \in \nhd_{\start}$ that indicates that some strategies in $\strats$ are not admissible as initial states.
Going back to the two archetypal examples of \eqref{eq:FTRL}, \cref{MWU,Projection}, there is a dichotomy in the properties of the corresponding mirror maps.
On the one hand, the kernel of the Euclidean/quadratic regularizer is differentiable on all of $[0,1]$.
On the other hand, the derivative of the kernel of the negative Shannon-entropy goes to $-\infty$ as $x$ goes to $0$.
This means that in the latter the boundaries are off the limits and inevitably some initial conditions do not belong in $\im\mirror$.
We discuss this dichotomy extensively in \cref{ss:steep vs non-steep}.
\end{remark}

\subsection{Main Results}%
\label{sec:statements}

We are now in a position to state our main results.
The informal version is as follows.

\begin{center}
\begin{tcolorbox}[enhanced,width=0.99\textwidth, drop fuzzy shadow southwest, 
                    boxrule=0.4pt,colframe=black!80!black,colback=black!10]\label{main theorem}
\vspace{0.1cm}
\centering
\textbf{Main Theorem.}
Suppose that Assumptions \ref{asm:bias}\textendash \ref{asm:degen} hold.
Then:
\\
$\eq$ is a strict \acl{NE}
	$\iff$
$\eq$ is stochastically asymptotically stable under \eqref{eq:FTRL}
\end{tcolorbox}
\end{center}

Formally, we get the following precise statements and corollaries for the specific feedback models described in \cref{sec:feedback}.

\begin{theorem}
\label{theorem:if}
Let $\eq\in\strats$ be a strict \acl{NE} of $\fingame$.
If \eqref{eq:FTRL} is run with inexact payoff feedback satisfying Assumptions \ref{asm:bias} and \ref{asm:variance}, then $\eq$ is stochastically asymptotically stable.
\end{theorem}

\begin{theorem}
\label{theorem:only-if}
Let $\eq$ be a mixed \acl{NE} of $\fingame$.
If \eqref{eq:FTRL} is run with inexact payoff feedback satisfying assumption \ref{asm:degen}, then $\eq$ is not stochastically asymptotically stable.
\end{theorem}

\begin{corollary}
Suppose that \eqref{eq:FTRL} is run in a generic game with oracle-based feedback as in \cref{semibandit feedback}
and a sufficiently small step-size $\step_{\run}$ with $\sum_{\run} \step_{\run}^{2} < \infty$.
Then, a \acl{NE} is stochastically asymprotically stable if and only if it is strict.
\end{corollary}

\begin{corollary}
Suppose that \eqref{eq:FTRL} is run in a generic game with bandit feedback as in \cref{bandit feedback}
and sufficiently small step-size and explicit exploration paramters with $\sum_\run \step_\run^2/\epar_\run <\infty$, $\sum_\run \step_\run \epar_\run <\infty$.
Then, a \acl{NE} is stochastically asymprotically stable if and only if it is strict.
\end{corollary}

These results \textendash\ and, in particular, the implications for the bandit case \textendash\ provide a learning justification to the abundance of arguments that have been made in the refinement literature against selecting mixed \aclp{NE} \citep{vD87,FT91,DF90}.
In the rest of our paper, we present an outline of the main proof ideas and defer the details to the appendix.

\section{Our Techniques}
\label{sec:Proof Techniques}
\subsection{The Stochastic Asymptotic Stability of Strict Nash Equilibria}\label{sec:Stability:Proof Techniques}
At a high level, the standard tool in FTRL dynamics for questions pertaining to asymptotic stability of strict \aclp{NE} is the construction of a potential \textendash\ or \emph{Lyapunov} \textendash\ function.
However, the analysis and the underlying structural results are considerably more involved when we shift from the continuous dynamics to discrete algorithms and more importantly in a stochastic framework with incomplete feedback information. Still, to build  intuition we first recall the continuous and deterministic analogue. 

\subparagraph{The continuous-time case.}
In prior work \citep{HS75,Wei00,Har11}, multiple instantiations of Bregman functions, like the KL-divergence  have been employed as a potent tool for understanding  \emph{replicator \& population dynamics}, which are the continuous analogues of MWU/EW (\cref{MWU}). 
Unfortunately, Bregman functions are  insufficient  to cover the full spectrum of regularizers studied in this work. %
 This limitation has been sidesteped in \cite{MS16} by exploiting the information of the dual space $\dspace$ of the payoff scores,  via the Fenchel coupling: %
\begin{equation}
\fench_\hreg\parens{x,y}  = \hreg\parens{x} + \hreg^*\parens{y} -\inner{y}{x} \text{ for all } x\in\pspace, y\in\dspace
\end{equation}   
where $\hreg^*: \dspace \to \R$ is the convex conjugate
of $\hreg$: $\hreg^*\parens{y} = \sup_{x\in\pspace}\braces{\inner{y}{x} -\hreg\parens{x}}$. 
Indeed, %
$\fench_\hreg\parens{\eq,y}\ge 0$ where equality holds if and only if $\eq=Q(y)$ (\cref{prop:Prop Fenchel}).
Therefore,  for the continuous FTRL dynamics $\dot{y}(t)=v(x(t)), x(t)=Q(y(t))$, it remains to show that  the time derivative of the Lyapunov-candidate-function $L_{\eq}(y(t))=\fench_\hreg\parens{\eq,y(t)}$ is negative. %
This last key ingredient for the strict Nash equilibria is derived by their \emph{variational stability} property. 
\noindent Formally, a point $\eq$ is \emph{variationally stable} if there exists a neighborhood $U$ of $\eq$ such that 
\begin{equation}\label{variational stability-main}
     \inner{\payv\parens{x}}{x-\eq}\leq 0 \text{ for all  }x\in U \tag{VS}
\end{equation}
with equality if and only if $x=\eq$. Roughly speaking, this property states that the payoff vectors are pointing ``towards'' the equilibrium in question since in a neighborhood of $\eq$, it strictly dominates over all other strategies. Thus by applying the chain rule, \eqref{variational stability-main} implies that ${\mathrm{d} L_{\eq}(y(t))}/{\mathrm{d} t}\leq 0$
\footnote{Analytically, $\dfrac{\mathrm{d} L_{\eq}(y(t))}{\mathrm{d} t}=\dfrac{\mathrm{d}{\hreg^*\parens{y(t)}}}{\mathrm{d} t} -\inner{\dot{y}(t)}{\eq}= \inner{\dot{y}(t)}{\nabla\hreg^*\parens{y}} -\inner{\dot{y}(t)}{\eq}= \inner{\payv\parens{x(t)}}{x(t)-\eq}\leq 0.
$%
}.
\noindent Given their usefulness also in the discrete time stochastic case, we present all the aformentioned properties in detail 
in the paper’s supplement (\cref{ss:Bregman}-\ref{ss:variational stability}).

\subparagraph{The discrete time.}
The core elements of the continuous time proof do not trivially extend to the discrete time case.
Even though  we are not able to show that $\parens{\fench_\hreg\parens{\eq,\dstate_\subround}}_{\subround=1}^{\infty}$ is a decreasing sequence, due to the discretization and the uncertainty involved, we prove that $ \fench_\hreg\parens{\eq,\dstate_\subround} \to 0$. This immediately implies that FTRL algorithm converges to $\eq$, since from \cref{prop:Prop Fenchel} $\fench_\hreg\parens{\eq,\dstate_\subround}\ge \frac{1}{2\stcon_\hreg}\norm{\eq-\pstate_\subround}$.

To exploit again the Fenchel coupling as a Lyapunov function, successive differences have to be taken among 
$\fench_\hreg\parens{\eq,\dstate_{\round+1}}, \ldots, \fench_\hreg\parens{\eq,\dstate_{0}}$.
In contrast to the continuous time analysis, since the chain rule no longer applies, we can only do a second order Taylor expansion of the Fenchel coupling. 
Additionally, let us recall that in our stochastic feedback model, the payoff vector $\npayv_\round= \payv\parens{\pstate_\round}+\noise_\round + \bias_\round$ including possibly either random zero-mean noise or systematic biased noise.
Combining \cref{prop:Prop Fenchel}, definition of $\npayv_\round$ and \eqref{eq:FTRL}, we can create the following upper-bound of Fenchel coupling at each round:
\begin{equation}\label{fenchel-telescopic-main}
     \fench_\hreg\parens{\eq,\dstate_{\round+1}} \leq \fench_\hreg\parens{\eq,\dstate_0} + \sum_{\subround =0}^\round \step_\subround(\rvdrift_\subround+\rvnoise_\subround+\rvbias_\subround)  + \dfrac{1}{2\stcon_\hreg}\sum_{\subround =0 }^\round \step_\subround^2\norm{\npayv_\subround}_*^2
\tag{$\star$}
\end{equation}
\smallskip
where $\rvdrift_\subround=\inner{\payv\parens{\pstate_\subround}}{\pstate_\subround - \eq},
\rvnoise_\subround=\inner{{\noise_\subround}}{\pstate_\subround - \eq},
\rvbias_\subround=\inner{{\bias_\subround}}{\pstate_\subround - \eq}
$ are the related terms with the drift of the actual payoff, the zero-mean noise and the bias correspondingly.
When $\pstate_\round$ lies in a variationally stable region $\nhd_{\ref{variational stability-main}}$ of $\eq$,
the first-order term of $\rvdrift_\subround$, which also appears in the continuous time, corresponds actually to the negative ``drift'' of the variational stability which attracts Fenchel coupling to zero. 

Having settled the basic framework, we split the proof sketch of \Cref{theorem:if} into two parts: \emph{stochastic stability} \& \emph{convergence}. Our analysis relies heavily on tools from the convex analysis and martingale limit theory to control the influence of the stochastic terms in the aforementioned bound.
\smallskip

\noindent\textbf{Step 1: \emph{Stability}.} Let $\nhd_\eps = \braces{x: \breg_\hreg\parens{\eq,x}<\eps}$ and $\nhd_\eps^* = \left\{y\in\dspace: \fench_\hreg\parens{\eq,y} < \eps\right\}$ be the $\eps-$sublevel sets of Bregman function and  Fenchel coupling respectively. Our first observation is that for all ``natural'' decomposable regularizers, it holds  the so-called ``reciprocity condition'' (\cref{prop:rec Bregman,prop:rec Fencel}): essentially, this posits that $\nhd_\eps$ and $Q(\nhd_\eps^*)$ are neighborhoods of $\eq$ in $\strats$. 
Additionally, since $\fench_\hreg\parens{\eq,y} = \breg_\hreg\parens{\eq,x}$ whenever $\mirror\parens{y}=x$ and $\supp\parens{x}$ contains $\supp\parens{\eq}$, from \cref{prop:Prop Fenchel}, it holds that  $\mirror\parens{\nhd_\eps^*}\subseteq \nhd_\eps$ and $\mirror^{-1}\parens{\nhd_\eps}= \nhd_\eps^*$. Thus, we conclude that whenever $y\in \nhd_\eps^*$, $x=\mirror\parens{y}\in \nhd_\eps$.

To proceed, fix a confidence level $\conlevel$ and $\eps$ sufficiently small such that \eqref{variational stability-main} holds for all $x \in \nhd_{\eps}$.
Using Doob's maximal inequalities  %
for (sub)martingales (\cref{Max inequality martingales,Max inequality submartingales}) we can prove that  with probability at least $1-\delta$,
\begin{enumerate*}[(a)]\item\label{item:noise} $\{\sum_{\subround =0}^\round \step_\subround\rvnoise_\subround\}$, \item \label{item:bias}$
\{\sum_{\subround =0}^\round \step_\subround\rvbias_\subround\}$ and \item \label{item:discrete-error}$ \{\frac{1}{2\stcon_\hreg}\sum_{\subround =0 }^\round \step_\subround^2\norm{\npayv_\subround}_*^2\}$\end{enumerate*} are less than $\eps/4$ for all $\round\ge 0$. 
For concision, we defer the full proof to the supplement of the paper in \cref{sec:deferred proofs-of main lemmas}. For the rest of this part, we condition on this event and rewrite \eqref{fenchel-telescopic-main} as $\fench_\hreg\parens{\eq,\dstate_{\round+1}} < \sum_{\subround =0}^\round \step_\subround\rvdrift_\subround+\eps$.
\medskip

\noindent Following the definition of stability (\cref{stability}),  we prove inductively that if  $\pstate_0$ belongs a smaller neighborhood, namely if  $\pstate_0 \in \nhd_{\eps/4}\cap \im\mirror$, then $\pstate_\round$ never escapes $\nhd_\eps$, $\pstate_\round\in \nhd_\eps$ for all $\round\geq 0$.

\begin{itemize}[topsep=.5ex,leftmargin=\parindent]
\setlength{\itemsep}{0pt}
\setlength{\parskip}{.2ex}
    \item  Induction Basis/Hypothesis: Since $\pstate_0 \in \nhd_{\eps/4}\cap \im\mirror$, apparently $\fench_\hreg\parens{\eq,\dstate_0} <\eps/4$ and $\pstate_0 \in \nhd_{\eps}$. Assume that $\pstate_\subround \in \nhd_{\eps}$ for all $0\leq \subround\leq \round$. 
    \item  Induction Step: We will prove that $\dstate_{\round +1}\in \nhd_{\eps}^*$ and consequently $\pstate_{\round+1}\in \nhd_{\eps}$. 
    Since $\nhd_{\eps}$ is a neighborhood of $\eq$ in which \eqref{eq:variational stability} holds  we have that
   $\rvdrift_\subround\leq 0$ for all  $0\leq\subround\leq\round $. Consequently $\fench_\hreg\parens{\eq,\dstate_{\round+1}} < \eps$
    which implies that $\dstate_{\round+1} \in \nhd_{\eps}^*$  or equivalently $\pstate_{\round+1} \in \nhd_{\eps}$.  
\end{itemize}
\smallskip

\noindent\textbf{Step 2: \emph{Convergence}.} A tandem combination of stochastic Lyapunov and variational stability is the following lemma:
\begin{lemma}[Informal statement of \cref{lem:strict remain to neighborhood}]
Let $\eq\in \pures$ be a strict Nash equilibrium. If $\curr[\pstate]$ does not exit a neighborhood $R$ of $\eq$, in which variational stability holds, then  there exists a subsequence $\pstate_{\round_\subround}$ of $\state_\round$ that converges to $\eq$ almost surely.
\end{lemma}
\noindent Indeed, if $\state_\round$ is entrapped in a variationally stable region $\nhd_\eps$ of $\eq$ without converging to $\eq$, we can show that $\sum_{\subround=0}^\infty \step_\subround \rvdrift_k \to -\infty$, while comparatively by the law of the large numbers for martingales (\cref{law of large numbers}), the contribution of \ref{item:noise},\ref{item:bias},\ref{item:discrete-error}  is negligible. Thus, in limit \eqref{fenchel-telescopic-main} implies that $0\leq \liminf \fench_\hreg\parens{\eq,\dstate_{\round}}\le  -\infty $, which is a contradiction.

Our final ingredient to complete the proof is that $\parens{\fench_\hreg\parens{\eq,\dstate_\subround}}_{\subround=1}^{\infty}$ behaves like an almost supermartingale when it is entrapped in a variationally stable region $\nhd_\eps$ of $\eq$. 
So, by convergence theorem for (sub)-martingales (\cref{Doob's convergence}),  $\parens{\fench_\hreg\parens{\eq,\dstate_\subround}}_{\subround=1}^{\infty}$ actually converges to a random  finite variable. Inevitably though, $\liminf_{\round\to\infty}\fench_\hreg\parens{\eq,\dstate_\round} = \lim_{\round\to\infty} \fench_\hreg\parens{\eq,\dstate_\round} = 0$ and by \cref{prop:Prop Fenchel}, $\mirror\parens{\dstate_\round} = \curr \to \eq$.

\subsection{The Stochastic Instability of Mixed Nash Equilibria}\label{sec:Instability:Proof Techniques}
For the proof of \cref{theorem:only-if}, it is worth mentioning that in this case stability fails for any choice of step-size. We start by focusing on the assumption of non-degeneracy \ref{asm:degen} of theorem's statement.%
\begin{itemize}[topsep=.5ex,leftmargin=\parindent]
\setlength{\itemsep}{0pt}
\setlength{\parskip}{.2ex}
\item From a game-theoretic perspective, \ref{asm:degen} actually demands that with non-zero probability, %
when players receive the payoffs corresponding to pure strategy profiles, there exists at least one player for whom at least two strategies of the equilibrium have distinct payoff signal. %
Note that if for each player, the payoffs corresponding to two different strategies of $\supp (\eq)$ were all equal \footnote{ when all other players' also employ strategies of the equilibrium} immediately implies a non-generic game with pure \aclp{NE}.
\item To illustrate this assumption in our generic feedback model, suppose that this error term $\error_{\run}$ is  standard normal random noise $\xi_{\run}$.
Indeed, the requirement of \ref{asm:degen} is satisfied since $\probof{\abs{ \payv_{\play, a}\parens{\pstate_{\run}} +\xi_{\play a,\run}-\payv_{\play,b}\parens{\pstate_{\run}}-\xi_{\play b,\run}}\geq 1/|\players| }> 1-\bigoh\left(\exp(-1/|\players|^2)\right)$. Such kind of property can be derived actually for any per-coordinate independent noise since actually the event of two independent coordinates to be exactly equal has zero measure.
\end{itemize}
For the bandit models \cref{semibandit feedback}, \cref{bandit feedback} of the previous section, we show that \ref{asm:degen} is satisfied in \cref{cor:non zero payoff,cor:payoffs bounded} of \cref{appendix instability}. 

Moving on to the proof of \Cref{theorem:only-if}, we start our analysis by connecting the difference of the payoff signal between two pure strategies, with the difference of the changes in the output of the regularizers' kernels, $\theta_\play$:  
\begin{lemma}[Informal Statement of \Cref{lem:Useful Expression}]\label{Main:Useful Expression}
Let $\pstate_{\play,\round}$ be the sequence of play in \eqref{eq:FTRL} \ie $\pstate_{\play,\round} =\mirror\parens{\dstate_{\play,\round}}\in\pspace_\play$ of player $\play\in\players$; and for some round $\round \geq 0$ let $a,b\in\supp\parens{\pstate_{\play,\round}}$ be two pure strategies of player $\play\in\players$. Then 
 it holds:
\begin{equation*}
   \parensnew{{\theta}'_\play\parens{\pstate_{\play a,\round+1 }} - {\theta}'_\play\parens{\pstate_{\play a,\round }}} - \parensnew{{\theta}'_\play\parens{\pstate_{\play b,\round+1}} - {\theta}'_\play\parens{\pstate_{\play b,\round}}}  = \step_\round\parens{\npayv_{\play a,\round}-\npayv_{\play b,\round}}
\end{equation*}
\end{lemma}

\noindent To proceed with the proof of \cref{theorem:only-if} assume ad absurdum that a mixed \acl{NE} $\eq$ is stochastically asymptotically stable.
Since $\eq$ is mixed, there exist $a,b\in\supp(\eq)$.
Second, the stochastic stability implies that for all $\eps,\conlevel >0$ if $\pstate_0$ belongs to an initial neighborhood $\nhd_{\eps}$, then $\norm{\curr -\eq} <\eps$ for all $\round \geq 0$, with probability at least $1-\conlevel$.
Third, by the triangle inequality for two consecutive instances of the sequence of play $\pstate_{\play,\round},\pstate_{\play,\round+1}$ for any player $\play\in\players$ it holds:
\begin{equation}\label{eq:closeness-per-round}
    |\pstate_{\play a,\round+1 } -\pstate_{\play a,\round }|+|\pstate_{\play b,\round+1 } -\pstate_{\play b,\round }|< \bigoh(\eps) \text{ with probability } 1-\conlevel
\end{equation}
Consider $\eps$ sufficiently small, such that the probabilities of the strategies that belong to the support of the equilibrium are bounded away from $0$, for all the points of the neighborhood. Since $\theta_\play$ is continuously differentiable in $(0,1]$, the differences described in \cref{Main:Useful Expression} are bounded from $\bigoh\parens{\eps}$ due to \cref{eq:closeness-per-round}. Thus, if the sequence of play $\curr$ is contained to an  $\eps-$neighborhood of $\eq$, then the difference of the feedback, for any player $\play\in\players$, to two strategies of the equilibrium is $\bigoh\parens{\eps/\step_\round}$ with probability at least $1-\conlevel$:
\[
 \probof{
\abs{\npayv_{\play a,\round} - \npayv_{\play b,\round}}= \bigoh\parens{\eps/\step_\round}\given\filter_\round} \geq 1-\conlevel
\]
However, from assumption \ref{asm:degen} for a fixed round $\round$ and some player $\play\in\players$, there exist $\noiselb,\problb>0$ such that:
\(
    \probof{\abs{\npayv_{\play a,\round} - \npayv_{\play b,\round}}\geq \noiselb\given\filter_\round} = \problb>0.
\)
Thus by choosing $\eps=\bigoh(\noiselb \step_\round )$ and $\delta=\problb/2$, we obtain a contradiction and our proof is complete.

\section{Discussion}

The equivalence between strict \aclp{NE} and stable attracting states of feedback-limited \eqref{eq:FTRL} implies that any equilibrium that exhibits payoff-indiffirence between different strategies is inherently unstable.
This fragility has already been remarked from an epistemic viewpoint \cite{vD87}, and our results provide a complementary justification based on realistic models of learning.

In the converse direction, the generality of the feedback models considered also provides a template for proving stochastic asymptotic stability results in more demanding learning environments.
A particular case of interest arises in online ad auctions where payoffs are observed with delay (or are dropped completely):
depending on the delay, the estimation of the player's payoff could exhibit a bias relative to the sampling strategy, and our generic conditions provide an estimate of how large the delays can be before convergence breaks down.
This opens the door to an array of fruitful research directions that we intend to pursue in the future.

%
%
%

%----------------------------------------------------------------------
%%% ACKS
%----------------------------------------------------------------------
%\acks{We thank a bunch of people and funding agency.}

%\clearpage

%**********************************************************************
%***    BIBLIOGRAPHY
%**********************************************************************
%\bibliographystyle{alpha}
\bibliography{bibtex/IEEEabrv,bibtex/Bibliography-PM}

%**********************************************************************
%***    APPENDICES
%**********************************************************************
\numberwithin{lemma}{section}		% for numbering  in the appendix
\numberwithin{proposition}{section}		% for numbering  in the appendix
\numberwithin{equation}{section}		% for numbering in the appendix
\newpage
\appendix

%----------------------------------------------------------------------
%%% APP: STABILITY
%----------------------------------------------------------------------
\section{Proof of stability of strict Nash equilibria}\label{appendix stability}
 Looking at the continuous analogues of FTRL algorithms, the standard methodology leverages  potential-Lyapunov arguments. However, in the discrete case multiple intrinsic challenges arise especially in the presence of uncertainty. In the prototypical example of MWU/EW (\cref{MWU}) the standard potential function is the Kullback-Leibler divergence. Thus, a natural candidate for the generalization from MWU/EW to any FTRL algorithm would be the Bregman divergence. Unfortunatelly, Bregman divergence does not always capture the actual behavior of FTRL algorithms in the boundary of the simplex. Hence, in order to trace the information entailed in the dual space, where FTRL algorithms truly evolve, we leverage the Fenchel coupling. In the first three subsections we present the main properties of Bregman divergence (\cref{ss:Bregman}),  the dichotomy among different regularizers (\cref{ss:steep vs non-steep}) and then Fenchel coupling and its properties (\cref{ss:Fenchel}). Then we explore a structural property of strict \aclp{NE}, namely variational stability (\cref{ss:variational stability}). In the last section before we present our proofs we introduce some notions from martingales limit theory. With these last two sections we have established all the necessary machinery to bound the influence of the uncertainty in the behavior of the algorithm and finally prove the stability result.

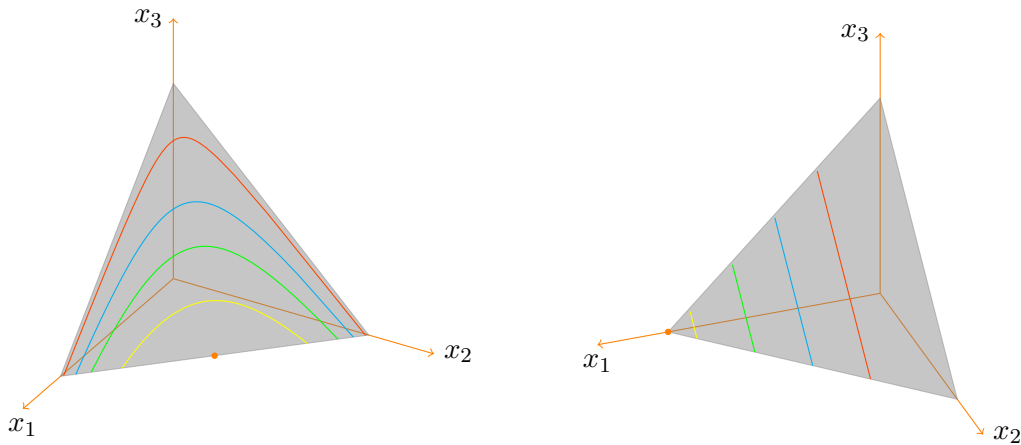
\begin{figure}[h!]
\usetikzlibrary{3d}
\tdplotsetmaincoords{60}{120}
\begin{tikzpicture}[tdplot_main_coords]
    \def\laxis{4}
    \def\ltriangle{3}
    \begin{scope}[->,orange,tdplot_main_coords]
        \draw (0,0,0) -- (\laxis,0,0) node [below] {\textcolor{black}{$x_1$}};
        \draw (0,0,0) -- (0,\laxis,0) node [right] {\textcolor{black}{$x_2$}};
        \draw (0,0,0) -- (0,0,\laxis) node [left] {\textcolor{black}{$x_3$}};
    \end{scope}
    \filldraw [opacity=.45,gray,tdplot_main_coords] (\ltriangle,0,0) -- (0,\ltriangle,0) -- (0,0,\ltriangle) -- cycle;
    \filldraw[orange,tdplot_main_coords] (1.5,1.5,0.0) circle (1pt);

%
%\draw (\ltriangle*(\outloop),\ltrianle*(1-\outloop),0)}
%\foreach \points\in{\outloop+0.01,\outloop+0.02,...,1-\outloop}{
%--(\ltriangle*\points,\ltriangle*\outloop*(1-\outloop)/\points,\ltriangle*(1-\points%-\outloop*(1-\outloop)/\points))}

\def\c{0.01*(1-0.01)}
\draw [OrangeRed,tdplot_main_coords]({\ltriangle*0.01}, {\ltriangle*(1-0.01)}, 0)
\foreach \t in {0.02, 0.03 ,..., 0.99}
{   --( {\ltriangle*\t},{\ltriangle*(\c/\t)}, {\ltriangle*(1-\t-\c/\t)}  )
};
\def\c{0.05*(1-0.05)}
\draw ({\ltriangle*0.05} ,{\ltriangle*(1-0.05)},0)[cyan]
\foreach \t in {0.07, 0.08 ,..., 0.95}
{   --( {\ltriangle*\t},{\ltriangle*(\c/\t)}, {\ltriangle*(1-\t-\c/\t)}   )
};
\def\c{0.1*(1-0.1)}
\draw ({\ltriangle*0.1} ,{\ltriangle*(1-0.1)},0)[green]
\foreach \t in {0.11, 0.12 ,..., 0.9}
{   --( {\ltriangle*\t}, {\ltriangle*(\c/\t)}, {\ltriangle*(1-\t-\c/\t)}   )
};
\def\c{0.2*(1-0.2)}
\draw ({\ltriangle*0.2} ,{\ltriangle*(1-0.2)},0)[Yellow]
\foreach \t in {0.21, 0.22 ,..., 0.8}
{   --( {\ltriangle*\t}, {\ltriangle*(\c/\t)} , {\ltriangle*(1-\t-\c/\t)})
};
\def\transferx{-10}
\def\transfery{0}
\def\transferz{-2.2}
    \tdplotsetrotatedcoords{0}{0}{-40}
    \begin{scope}[->,orange,tdplot_rotated_coords]
        \draw (0+\transferx,0+\transfery,0+\transferz) -- (\laxis+\transferx,0+\transfery,0+\transferz) node [below] {\textcolor{black}{$x_1$}};
        \draw (0+\transferx,0+\transfery,0+\transferz) -- (0+\transferx,\laxis+\transfery,0+\transferz) node [right] {\textcolor{black}{$x_2$}};
        \draw (0+\transferx,0+\transfery,0+\transferz) -- (0+\transferx,0+\transfery,\laxis+\transferz) node [left] {\textcolor{black}{$x_3$}};
    \end{scope}
 
    \filldraw [opacity=.45,gray,tdplot_rotated_coords] (\ltriangle + \transferx,0+\transfery,0+\transferz) -- (0+\transferx,\ltriangle+\transfery,0+\transferz) -- (0+\transferx,0+\transfery,\ltriangle+\transferz) -- cycle;
    \filldraw[orange,tdplot_rotated_coords] (1*\ltriangle+\transferx,0+\transfery,0+\transferz) circle (1pt);
\def\c{0.3}
\draw[OrangeRed,tdplot_rotated_coords] ({\ltriangle*\c +\transferx},{ (1-\c)*\ltriangle +\transfery} ,{0+\transferz})
\foreach \t in {0.01, 0.02 ,..., 0.7}
{   --( {\ltriangle*\c+\transferx}, {\ltriangle*(1-\t-\c) +\transfery}, {\ltriangle*(\t)+\transferz}  )
};
\def\c{0.5}
\draw[cyan,tdplot_rotated_coords] ({\ltriangle*\c +\transferx},{ (1-\c)*\ltriangle +\transfery} ,{0+\transferz})
\foreach \t in {0.01, 0.02 ,..., 0.5}
{   --( {\ltriangle*\c+\transferx}, {\ltriangle*(1-\t-\c) +\transfery}, {\ltriangle*(\t)+\transferz}  )
};
\def\c{0.7}
\draw[green,tdplot_rotated_coords] ({\ltriangle*\c +\transferx},{ (1-\c)*\ltriangle +\transfery} ,{0+\transferz})
\foreach \t in {0.01, 0.02 ,..., 0.3}
{   --( {\ltriangle*\c+\transferx}, {\ltriangle*(1-\t-\c) +\transfery}, {\ltriangle*(\t)+\transferz}  )
};
\def\c{0.9}
\draw[Yellow,tdplot_rotated_coords] ({\ltriangle*\c +\transferx},{ (1-\c)*\ltriangle +\transfery} ,{0+\transferz})
\foreach \t in {0.01, 0.02 ,..., 0.1}
{   --( {\ltriangle*\c+\transferx}, {\ltriangle*(1-\t-\c) +\transfery}, {\ltriangle*(\t)+\transferz}  )
};
\end{tikzpicture}
\caption{The level sets of KL-divergence}
\label{level sets of KL}
\end{figure}

\subsection{Bregman divergence}\label{ss:Bregman}
 Bregman divergence provides a way to measure the distance of two points that belong to the simplex. Its properties render it a useful tool to prove convergence results. Below we state its definition and prove these properties that would be crucial in the establishment of our proof. Given a fixed point $p\in\pspace$ then the Bregman divergence of a function $\hreg$ is defined for all points $x\in\pspace$ as
\begin{equation}\label{eq:Bregman}
\breg_\hreg\parens{p,x} = \hreg\parens{p} - \hreg\parens{x} - \hreg'\parens{x;p-x} \text{ for all }p,x \in\pspace
\end{equation}
where $\hreg'\parens{x;p-x}$ is the one-sided derivative
\begin{equation}
    \hreg'\parens{x;p-x}\equiv \lim_{t\to 0^+} t^{-1} \bracks{\hreg\parens{x+t\parens{p-x}} -\hreg\parens{x}}
\end{equation}
Notice that this definition of the Bregman divergence permits to work also with points on the boundary. It is possible that the limit of $\breg_\hreg$ attains the value of $+\infty$ if $\hreg'\parens{x;p-x} = -\infty$, as  $x\to p$, where $p$ is a point of the boundary. However, the condition below ensures that this is not the case.
\begin{equation}\label{eq:Bregman reciprocity}\tag{Reciprocity}
    \breg_\hreg\parens{p;x} \to 0 \text{ whenever }x\to p
\end{equation}
This is known as the reciprocity condition. What this property actually means is that the sublevel sets of $\breg\parens{p,\cdot}$ are neighborhoods of $p$. This is illustrated in \cref{level sets of KL}, when the function employed is the negative Shannon-entropy and the induced Bregman divergnce the Kullback–Leibler divergence. %
Notice that for most decomposable functions $\hreg$, this property holds. Below we present a proof of this statement.
\begin{proposition}\label{prop:rec Bregman}
If $\hreg\parens{x} = \sum_i\theta\parens{x_i}$, with $\theta$ having the properties described in \eqref{h properties} and furthermore it holds that $\theta'\parens{x} = o\parens{1/x}$ for $x$ close to $0$, then $\breg_\hreg\parens{p;x}\to 0 $ whenever $x\to p$ for all $x,p\in\pspace$.
\end{proposition}
\begin{proof}
It is sufficient to prove that $\lim_{x\to 0}\parens{\theta\parens{0} - \theta\parens{x} - \theta'\parens{x}\parens{0-x}} = 0$. The difference of the first two terms is obviously gives zero. Now, for the last term notice that if $\theta'\parens{x} = o\parens{1/x}$ for $x$ close to $0$, then $\lim_{x\to0}x\theta'\parens{x} = 0$ and the proof is completed.
\end{proof}
\noindent Additionally, Bregman divergence satisfies the properties described below.
\begin{proposition}\label{prop:Prop Bregman}
Let $\hreg$ be a $\stcon$-strongly convex function defined on the simplex $\pspace = \Delta\parens{\pures}$, that has the properties described in \ref{h properties} and let $\Delta_{p}$ be the union of the relative interiors of the faces of $\pspace$ that contain $p$ \ie
\begin{equation}
    \Delta_{p} = \braces{x\in\pspace : \supp\parens{p}\subseteq\supp\parens{x}} = \braces{x\in\pspace: x_a >0 \text{ whenever }p_a>0}
\end{equation}
Then
\begin{enumerate}
    \item $\breg_\hreg\parens{p,x}< \infty$ whenever $x\in\Delta_{p}$.
    \item $\breg_\hreg\parens{p,x}\geq 0$ for all $x\in\pspace$, with equality if and only if $p = x$, more particularly
    \begin{equation}
        \breg_\hreg\parens{p,x} \geq \dfrac{1}{2}\stcon\norm{x-p}^2\text{ for all }x\in\pspace
    \end{equation}
\end{enumerate}
\end{proposition}
\begin{proof} For the first part, if $x\in\Delta_p$ then $\hreg\parens{x+t\parens{x-p}}$ is finite and smooth in a neighborhood of $0$ and thus $\breg\parens{p,x}$ is also finite.\\
The second  part of the proposition, let $z=x-p$ then strong convexity yields
\begin{align*}
    \hreg\parens{x + tz}\leq t\hreg\parens{p} +\parens{1-t}\hreg\parens{x} - \dfrac{1}{2}\stcon t\parens{1-t}\norm{x-p}^2&\\
    t^{-1}\parens{\hreg\parens{x + tz} - \hreg\parens{x}} \leq \hreg\parens{p} -\hreg\parens{x}  - \dfrac{1}{2}\parens{1-t}\stcon \norm{x-p}^2&\\
    \hreg\parens{p} -\hreg\parens{x} -  t^{-1}\parens{\hreg\parens{x + tz} - \hreg\parens{x}} \geq \dfrac{1}{2}\parens{1-t}\stcon \norm{x-p}^2&
\end{align*}
And by taking $t\to 0$, we obtain the result.
\end{proof}
We mention at this point that from \eqref{h properties}, since for each $\play\in\players$: $\inf_{\in[0,1]}\theta_\play''>0$, there exists $\stcon_\play >0$ such that for all $x,y\in[0,1]$ and $t\in[0,1]$
\begin{equation}
    \theta_\play\parens{tx+(1-t)y}\leq t\theta_\play\parens{x} + \parens{1-t}\theta_\play\parens{y} - \dfrac{\stcon_\play}{2}t\parens{1-t}\abs{x-y}^2
\end{equation}
\noindent In all the proofs
$\hreg$ symbolizes the aggregate function of all the regularizers \ie $\hreg\parens{x}=\sum_{\play}\hreg_{\play}\parens{x_{\play}}$, with strong convexity parameter $\stcon\equiv \min_{\play}\stcon_{\play}$.
\subsection{Steep vs non-steep}\label{ss:steep vs non-steep} In this section we elaborate in detail the dichotomy of the properties of different regularizers mentioned in \cref{stability remark}. As we mentioned players may have different regularizers $\hreg_\play$ employed in their choice maps $\mirror_\play\parens{y}=\argmax_{x\in\pspace_\play}\left\{\inner{x}{y}-\hreg_\play\parens{x}\right\}$. Depending on the regularizer chosen, FTRL dynamics may differ significantly.
To formally express this difference,  it is convenient to consider that $\hreg$ is an extended-real valued function $\hreg:\mathcal{V}\to\R\cup\braces{\infty}$ with value $\infty$ outside of the simplex $\pspace$. Then the subdifferential of $\hreg$ at $x\in\mathcal{V}$ is defined as:
\begin{equation}
\partial\hreg\parens{x} = \braces{y\in {\mathcal{V}}^* : \hreg\parens{{x}'}\geq \hreg\parens{x} +\inner{y}{{x}'-x} \; \forall {x}'\in \mathcal{V}}
\end{equation}
If $\partial\hreg\parens{x}$ is nonempty, then $\hreg$ is called subdifferentiable at $x\in\pspace$.  When $x\in\relint\parens{\pspace}$ then $\partial\hreg\parens{x}$ is always non-empty or $\relint\parens{\pspace}\subseteq \dom\partial\hreg\equiv \braces{x\in\pspace:\partial\hreg\parens{x}\neq \emptyset}$.  Notice that when the gradient of $\hreg$ exists, then its subgradient always contains it. With these in mind, we present a typical separation between the different regularizers. On the one hand,  \emph{steep} regularizers like the negative Shannon-entropy  become infinitely steep as $x$ approaches the boundary or $\norm{\nabla h\parens{x}}\to\infty$. On the other hand, \emph{non-steep} are everywhere differentiable, like the Euclidean, allowing the sequence of play to transfer between the different faces of the simplex. In the dual space of payoffs, steepness implies that the choice map is not surjective (since it cannot map all payoff vectors to points of the boundary), it is however injective (it maps a payoff vector plus a  multiple of $(1,1,\ldots,1)$ to the same strategy). Non-steep regularizers give rise to surjective maps, which are not injective, not even up to a multiple of $\parens{1,1,\hdots,1}$, to the boundary. Focusing on the more simple case of decomposable regularizers, the kernel of a steep one is differentiable on $(0,1]$ while for non-steep the kernel is differentiable in all of $[0,1]$. As a result, when a steep regularizer is employed the mirror map $\mirror:\dspace\to\pspace$ cannot return any point of the boundary. In other words, the points of the boundary are infeasible not only as initial conditions but also as part of the sequence of play.
\begin{remark}
This dichotomy is important for our analysis since we study the stochastic asymptotic stability of Nash equilibria, which may lie on the boundary, and we seek a neighborhood of initial conditions such that the equilibrium to be stable and attracting. Thus, instead of demanding the existence of a neighborhood $\nhd$ of an equilibrium $\eq$, such that whenever $X_0\in \nhd$, $\eq$ is stable and attracting; we demand the existence of a neighborhood $U$ of $\eq$ such that whenever $X_0\in \nhd\cap\im\mirror$ then $\eq$ is stable and attracting.
\end{remark}
\subsection{Polar cone}\label{ss:polar cone} The notion of the polar cone is tightly connected with the notion of duality.
Given a finite dimensional vector space $\mathcal{V}$, a convex set $\mathcal{C}\subseteq\mathcal{V}$ and a point $x\in\mathcal{C}$ the tangent cone $\tcone_{\mathcal{C}}\parens{x}$ is the closure of the set of all rays emanating from $x$ and intersecting $\mathcal{C}$ in at least one other point. The dual of the tangent cone is the polar cone $\pcone_{\mathcal{C}}\parens{x} = \braces{y\in\mathcal{V}^*:\inner{y}{z}\leq 0 \text{ for all }  z\in\tcone_{\mathcal{C}}\parens{x}}$. \\
When the under consideration convex set is the simplex of the players' strategies, the polar cone corresponding to the boundary  differs significantly from the one corresponding to the interior. Formally, the polar cone at a point $x$ of the simplex is
\begin{equation}
    \pcone\parens{x} = \braces{y\in\dspace: y_a \geq y_b \text{ for all } a,b\in\pures}\footnote{It is always $y_a=y_b$ whenever $a,b\in\supp\parens{x}$.}
\end{equation} An illustration of this is depicted in \cref{fig: polar cone}.
When \eqref{eq:FTRL} is run, the notion of the polar cone emerges from the choice map $\mirror:\dspace\to\pspace$, connecting the primal space of the strategies with the dual space of the payoffs. The proposition below presents this exact connection. 
\begin{proposition}\label{prop:y in subgradient} 
Let $\hreg$ be a strong convex regularizer that satisfies the properties described in \ref{h properties} and let $\mirror:\dspace\to\pspace$ be the induced choice map then %
\begin{enumerate}
    \item $x=\mirror\parens{y}\Leftrightarrow y\in\partial\hreg\parens{x}$
    \item $\partial\hreg\parens{x}= \nabla\hreg\parens{x} + \pcone\parens{x}$ for all $x\in\pspace$.
\end{enumerate}
\end{proposition}

\subsection{Fenchel coupling}\label{ss:Fenchel}
Even though Bregman divergence is a useful tool, \eqref{eq:FTRL} evolves in the dual space of payoffs. Thus dually to the above the Fenchel coupling\footnote{The term is due to \cite{MS16}.}  is defined, $\fench_\hreg:\pspace\times \dspace \to \R$
\begin{equation}
    \fench_\hreg\parens{p,y}  = \hreg\parens{p} + \hreg^*\parens{y} -\inner{y}{p} \text{ for all } p\in\pspace, y\in\dspace
\end{equation}
where $\hreg^*: \dspace \to \R$ is the convex conjugate of $\hreg$: $\hreg^*\parens{y} = \sup_{x\in\pspace}\braces{\inner{y}{x} -\hreg\parens{x}}$. 
The fenchel conjugate is  differentiable on $\dspace$ and it holds that
\begin{equation}
    \nabla\hreg^*\parens{y} = \mirror\parens{y} \text{ for all } y\in \dspace
\end{equation}
Fenchel coupling is also a measure that connects the primal with the dual space. As we mentioned above, \eqref{eq:FTRL} evolves in the dual space and thus we use Fenchel coupling to trace its convergence properties. As the next proposition states, whenever Fenchel coupling $\fench\parens{p,y}$ is bounded from above so does $\norm{\mirror\parens{y}-p}$. This proposition in its entity, is critical for our proof, since we first need to find a neighborhood $U$ of attractness (See \cref{stability}). For this step, Bregman divergence is  necessary in order to define the aforementioned neighborhood since $\norm{\mirror\parens{y}-p}< c$ for some constant $c$ is not necessarily a neighborhood of $p$ (See \cref{ss:steep vs non-steep}). 
\begin{proposition}\label{prop:Prop Fenchel}
Let $\hreg$ be a $\stcon$-strongly convex function on $\pspace$ and has the propertied described in \ref{h properties}. Let $p\in\pspace$, then
\begin{enumerate}
    \item $\fench_\hreg\parens{p,y}\geq \dfrac{1}{2}\stcon\norm{\mirror\parens{y}-p}^2$ for all $y\in\dspace$ and whenever $\fench_\hreg\parens{p,y}\to 0$, $\mirror\parens{y}\to p$.
    \item $\fench_\hreg\parens{p,y} = \breg_\hreg\parens{p,x}$ whenever $\mirror\parens{y}=x$ and $x\in\Delta_{p}$.
    \item $\fench_\hreg\parens{p,{y}'}\leq \fench_\hreg\parens{p,y}+\inner{{y}'-y}{\mirror\parens{y}-p} + \dfrac{1}{2\stcon}\norm{{y}'-y}_*^2$.
\end{enumerate}
\end{proposition}

\begin{remark}
Notice that the first part of the proposition is not implied by the second one, since it is possible that $\im\mirror = \dom\partial \hreg$ is not always contained in $\Delta_p$ (see \cref{ss:steep vs non-steep}). 
\end{remark}
\begin{proof}
For the first part, let $x=\mirror\parens{y}$ then $\hreg^*\parens{y} = \inner{y}{x} - \hreg\parens{x}$ 
\begin{equation}
    \fench_\hreg\parens{p,y} = \hreg\parens{p} - \hreg\parens{x} - \inner{y}{p-x}
\end{equation}
Since  $y\in\partial \hreg\parens{x}$ (\cref{prop:y in subgradient}), it is
\begin{equation}\label{eq 1}
    \hreg\parens{x+t\parens{p-x}}
    \geq \hreg\parens{x} + t\inner{y}{p-x}
\end{equation}
and by strong convexity of $\hreg$, we have 
\begin{equation}\label{eq 2}
    \hreg\parens{x+t\parens{p-x}}\leq t\hreg\parens{p} +\parens{1-t}\hreg\parens{x} -\dfrac{1}{2}\stcon t\parens{1-t}\norm{p-x}^2
    \end{equation}
Thus by combining \eqref{eq 1},\eqref{eq 2} and taking $t\to 0$ we get
\begin{equation}
    \fench_\hreg\parens{p,y} \geq \hreg\parens{p} -\hreg\parens{x}  -\hreg\parens{p} +\hreg\parens{x} + \dfrac{\stcon}{2}\norm{p-x}^2 \geq \dfrac{\stcon}{2}\norm{p-x}^2 
\end{equation}
For the second part of the proposition, notice that $x+t\parens{p-x}$ lies in the relative interior of some face of $\pspace$ for $t$ in a neighborhood of $0$ and thus $\hreg\parens{x+t\parens{p-x}}$ is smooth and finite. So, $\hreg$ admits a two-sided derivative along $x-p$ and since $y\in\partial\hreg\parens{x}$, $\inner{y}{p-x} = \hreg'\parens{x;p-x}$ and our claim naturally follows.\\
Finally for the last part of the proposition, we have
\begin{align*}
    \fench_\hreg\parens{p,y'} &= \hreg\parens{p} +\hreg^*\parens{y'}  -\inner{y'}{p}\\
    &\leq \hreg\parens{p} + \hreg^*\parens{y} + \inner{y'-y}{\nabla\hreg^*\parens{y}} +\dfrac{1}{2\stcon}\norm{y'-y}^2_*  -\inner{y'}{p}\\
    &= \fench_\hreg\parens{p,y} + \inner{y'-y}{\mirror\parens{y} - p}  +\dfrac{1}{2\stcon}\norm{y'-y}^2_*
\end{align*}
where the second inequality follows from the fact that $\hreg^*$ is $1/\stcon$ strongly smooth \cite{RW98}.
\end{proof}
In terms of Fenchel coupling our reciprocity assumption can be written as 
\begin{equation}\label{eq: Fenchel reciprocity}\tag{Reciprocity}
    \fench_\hreg\parens{p,y}\to 0 \text{ whenever } \mirror\parens{y}\to p
\end{equation}
Again for most of $\hreg$ decomposable, the assumption is turned into a property as we prove below.
\begin{proposition}\label{prop:rec Fencel} If $\hreg\parens{x} = \sum_i\theta\parens{x_i}$, with $\theta$ having the properties described in \eqref{h properties} and furthermore it holds that $\theta'\parens{x} = o\parens{1/x}$ for $x$ close to $0$, then $\fench_\hreg\parens{p,y}\to 0 $ whenever $\mirror\parens{y}\to p$ for all $p\in\pspace$.
\end{proposition}
\begin{proof}
Again it is sufficient to prove that whenever $\mirror\parens{y}=x\to 0$ then $\fench_\hreg\parens{p,y}\to 0$. Notice that from \cref{prop:Prop Fenchel} $\fench_\hreg\parens{p,y} =\breg_\hreg\parens{p,x}$ whenever $x=\mirror\parens{y}$ and $x\in\Delta_p$. Thus by \cref{prop:rec Bregman} $\mirror\parens{y}=x\to 0$ implies that $\fench_\hreg\parens{p,y}\to0$. 
\end{proof}
\subsection{Variational stability} 
\label{ss:variational stability}
\begin{definition}[Variational stability] A point $\eq\in\pspace$ is said to be \emph{variationally stable} if there exists neighborhood $\nhd$ of $\eq$ such that
\begin{equation}\label{eq:variational stability}\tag{VS}
    \inner{\payv\parens{x}}{x-\eq}\leq 0 \text{ for all } x\in \nhd
\end{equation}with  equation if and only if $x=\eq$.
\end{definition}
\noindent What this property actually states is that in a neighborhood of $\eq$, it strictly dominates over all other strategies. Interestingly, strict Nash equilibria hold this property:

\begin{proposition}\label{prop:strict variational stable}For finite games in normal form, the  following are equivalent:
\begin{enumerate}[i)]
    \item $\eq$ is a strict Nash equilibrium.
    \item $\inner{\payv\parens{\eq}}{z}\leq 0$ for all $z\in\tcone\parens{\eq}$ with equality if and only if z=0.
    \item $\eq$ is variationally stable.
\end{enumerate}
\end{proposition}
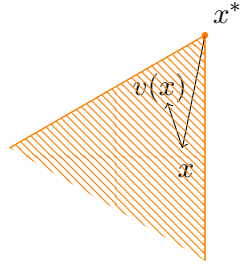
\begin{figure}[h!]
\centering
    \begin{tikzpicture}
\def\len{3}
\def\appr{1.732}
\draw[color = orange] (0,-\len)-- (0,0)--(-\len*\appr/2,-\len/2);
\filldraw[color = orange] (0,0) circle (1pt);
\node at (0.3,0.3) {$x^*$};
%\draw[color = cyan] (\len/2,\len*\appr/2)--(0,0)--(\len/2,-\len*\appr/2);
\draw[->] (0,0)--(-0.3,-1.5);
\node at (-0.3+0.05,-1.5-0.3) {$x$};
\draw[->] (-0.3,-1.5) -- (-0.3-0.2,-1.5+0.6);
\node at (-0.3-0.3,-1.5+0.8) {$\tiny v(x)$};
\filldraw[color = orange, pattern=north west lines,pattern color=orange] (0,-\len)-- (0,0)--(-\len*\appr/2,-\len/2);
\end{tikzpicture}
  \caption{\eqref{eq:variational stability} states that the payoff vectors are pointing "towards" the equilibrium }
\end{figure}
\begin{proof}

We will first prove that $i)\Rightarrow ii)\Rightarrow iii)\Rightarrow i)$.\\
$i)\Rightarrow ii)$ Since $\eq$ is a Nash equilibrium by definition it holds for each player $\play$ that 
\begin{equation}
  \inner{\payv\parens{\eq}}{x-\eq}\leq 0 \text{ for all }x\in\pspace
\end{equation}
For the strict part of the inequality, by definition of strict Nash equilibria it holds that $\inner{\payv_\play\parens{\eq}}{x_\play-\eq_\play} <0$ whenever $x_\play\neq \eq_\play$ and thus
\begin{equation}
    \inner{\payv\parens{\eq}}{z} = \sum_{\play=1}^{N}\inner{\payv_\play\parens{\eq}}{x_\play-\eq_\play} <0 \text{ if } x_\play\neq \eq_\play \text{ for some } \play \text { or }z\neq 0
\end{equation}
$ii)\Rightarrow iii)$ By definition of the polar cone, we have that $\payv\parens{\eq}$ belongs to the interior of $\pcone\parens{\eq}$\footnote{Indeed if it belonged to the boundary then the equality in $ii)$ would not hold only for $z=0$.}. Thus by continuity there exists some neighborhood of $\eq$ such that $\payv\parens{x}$ also belongs to the polar cone of $\pcone\parens{\eq}$ or $\eq$ is variationally stable.\\
$iii)\Rightarrow i)$ Assume now that $\eq$ is variationally stable but not strict, then there exist for some player $\play$ $a$,$b\in\pures_\play$ such that $\pay_{\play}\parens{a;\eq_{-\play}} = \pay_{\play}\parens{b;\eq_{-\play}}$. Then for $x_\play = \eq_\play +\lambda\parens{\bvec_a - \bvec_b}$ and $x_{-\play} =\eq_{-\play}$  we have
\begin{equation}
    \inner{\payv\parens{\eq}}{x -\eq} =  \inner{\payv_\play\parens{\eq}}{\lambda\parens{\bvec_a - \bvec_b}} = 0
\end{equation}
which is a contradiction.
\end{proof}
\subsection{Martingale limit theory}
Our analysis leverages tools from martingale limit theory. Below we first present a simple fact for the reader to keep in mind, followed by the main theorems that we utilize in the main body of our proofs
\begin{fact}\label{fact: submartingale}
Let $R_\round = \sum_{\subround =1}^\round r_\subround$, where $r_\subround$ is a positive random variable for all $\subround = 0,1,\hdots$ attached to the filtration $\filter_{\subround-1}$. Then $R_\round$ is a submartingale.
\end{fact}
\noindent We begin with the strong law of large numbers for martingale difference sequences:
\begin{theorem}\label{law of large numbers}
Let $R_\round =\sum_{\subround =1}^\round r_\subround$ be a martingale with respect to an underlying stochastic basis $\parens{\Omega,\filter,\parens{\filter_\round}_{\round=1}^\infty,\prob}$ and let $\parens{\tau_\round}_{\round=1}^\infty$ be a nondecreasing sequence of positive numbers with $\lim_{\round \to\infty}\tau_\round =  \infty$. If $\sum_{\round =1}^\infty \tau_\round^{-p}\exof{\abs{r_\round}^p\given\filter_{\round-1}}<\infty$ for some $p\in[1,2]$ almost surely, then
\begin{equation}
    \lim_{\round\to\infty}\tau_\round^{-1} R_\round = 0 \text{ almost surely}
\end{equation}
\end{theorem}
\noindent The second important result for our analysis is Doob's martingale convergence theorem:
\begin{theorem}\label{Doob's convergence}
If $R_\round$ is a submartingale that is bounded in $L_1$ (\ie $\sup_\round \exof{\abs{R_\round}}<\infty$), $R_\round$ converges almost surely to a random variable $R$ with $\exof{R}<\infty$.
\end{theorem}
\noindent Finally, we use the known as Doob's maximal inequality and one of its variants, presented below:
\begin{theorem}\label{Max inequality submartingales}
Let $R_\round$ be a non-negative submartingale and fix some $\eps > 0$. Then:
\begin{equation}
    \probof{\sup_\round R_\round \geq \eps}\leq \dfrac{\exof{R_\round}}{\eps}
\end{equation}
\end{theorem}
\begin{theorem}\label{Max inequality martingales}
Let $R_\round$ be a martingale and fix some $\eps >0$. Then:
\begin{equation}
    \probof{\sup_\round \abs{R_\round} \geq \eps}\leq \dfrac{\exof{R_\round^2}}{\eps^2}
\end{equation}
\end{theorem}
\noindent Proofs of all these results can be found in \cite{HH80}.
\subsection{Deferred Proof of  \cref{theorem:if}}
\label{sec:deferred proofs-of main lemmas}
In the following preliminary result, we focus on the case of \eqref{eq:FTRL} with payoff feedback as described in \cref{sec:feedback} and we show that if $\eq$ is a \emph{strict} Nash equilibrium, there exists a subsequence of $\parens{\pstate_\round}_{\round=0}^\infty$ that converges to it. In order to achieve this convergence result, it is necessary to assume that the sequence $\parens{\pstate_\round}_{\round=0}^\infty$ is contained in a neighborhood of $\eq$, in which \eqref{eq:variational stability} holds. %
Here, we  outline the basic steps below:
\begin{enumerate}
 \item[Step 0:] By contradiction, assume that there exists a neighborhood, in which $\pstate_\round$ is not contained for all sufficiently large $\round$ and assume without loss of generality that holds for all $\round = 0,1,\hdots$. 
  \item[Step 1:]We start by showing that the terms of the RHS of the third property described in \cref{prop:Prop Fenchel} are converging almost surely to finite values, except for one. This term, which is a consequence of $\eq$ being variational stable, goes to $-\infty$ as $\round\to\infty$ .
 \item[Step 2:] The next crucial observation is that the Fenchel coupling is bounded from below by $0$, thanks to the first property in \cref{prop:Prop Fenchel}, which gives us the contradiction.
\end{enumerate}

\begin{remark*}
For the interested reader, the assumption \ref{asm:variance},  $\sum_{\run} \step_{\run}^{2} \sbound_{\run}^{2} < \infty$, that we use in the preliminary lemma and in \cref{theorem:if} could be relaxed by using the H\"{o}lder inequality to $\sum_{\run} \step_{\run}^{1+q/2} \sbound_{\run}^{q}<\infty$ for any $q\in [2,\infty)$.
\end{remark*}

\begin{lemma}\label{lem:strict remain to neighborhood}
Let $\eq\in \pures$ be a strict Nash equilibrium. If \eqref{eq:FTRL} is run with payoff feedback of the type \eqref{eq:feedback}, that satisfies \ref{asm:bias}-\ref{asm:variance} and the sequence of play $\parens{\curr[\pstate]}_{\round=0}^\infty$ does not exit a neighborhood $\mathcal{R}$ of $\eq$, in which variational stability holds, then there exists a subsequence $\pstate_{\round_\subround}$ of $\state_\round$ that converges to $\eq$ almost surely.%, under the hypotheses: $\sum_{\subround=0}^\infty \step_\subround \to \infty$, $\sum_{\subround=0}^\infty \step_\subround^2 <\infty$.
\end{lemma}
\begin{proof}
Suppose that there exists a neighborhood $\nhd\subseteq \mathcal{R}$ of $\eq$ , such that $\curr[\pstate]\notin \nhd$ for all large enough $\round$. Assume without loss of generality that this is true for all $\round\geq 0$. Since variational stability holds in $R$, we have
\begin{equation}\label{formulation}
    \inner{\payv\parens{\strat}}{\strat -\eq} <0 \text{ for all } \strat\in \mathcal{R} ,\; \strat\neq \eq
\end{equation}
Furthermore, from \cref{prop:Prop Fenchel} we have that for each round $\round$:
\begin{equation}
    \fench_\hreg\parens{\eq,\dstate_{\round+1}} \leq \fench_\hreg\parens{\eq,\curr[\dstate]} + \step_\round \inner{\curr[\npayv]}{\curr - \eq} + \dfrac{1}{2\stcon}\step_\round^2 \norm{\curr[\npayv]}^2_*
\end{equation}
By applying the above inequality for all rounds from $1,...,\round$ and creating the telescopic sum we get
\begin{equation}\label{fenchel inequality}
    \fench_\hreg\parens{\eq,\dstate_{\round+1}} \leq \fench_\hreg\parens{\eq,\dstate_0} + \sum_{\subround =0}^\round \step_\subround\inner{\npayv_\subround}{\pstate_\subround - \eq} + \dfrac{1}{2\stcon}\sum_{\subround =0 }^\round \step_\subround^2\norm{\npayv_\subround}_*^2
\end{equation}
Remember that for the payoff vector holds that $$\npayv_\round= \payv\parens{\pstate_\round}+\noise_\round + \bias_\round$$
\noindent We now rewrite \eqref{fenchel inequality} 
\begin{equation}
\begin{aligned}
     \fench_\hreg\parens{\eq,\dstate_{\round+1}} \leq \fench_\hreg\parens{\eq,\dstate_0} &+ \sum_{\subround =0}^\round \step_\subround\inner{\payv\parens{\pstate_\subround}}{\pstate_\subround - \eq} + \sum_{\subround =0}^\round \step_\subround\inner{\noise_\subround}{\pstate_\subround -\eq}\\
     &+\sum_{\subround =0}^\round \step_\subround\inner{\bias_\subround}{\pstate_\subround -\eq} + \dfrac{1}{2\stcon}\sum_{\subround =0 }^\round \step_\subround^2\norm{\npayv_\subround}_*^2
\end{aligned}
\end{equation}
Let $\tau_\round = \sum_{\subround =0}^\round \step_\subround$ then 
\begin{equation}\label{fenchel inequality 2}
\begin{aligned}
     \fench_\hreg\parens{\eq,\dstate_{\round+1}} \leq \fench_\hreg\parens{\eq,\dstate_0} &+ \sum_{\subround =0}^\round \step_\subround\inner{\payv\parens{\pstate_\subround}}{\pstate_\subround - \eq} +\tau_\round\parens{\dfrac{\sum_{\subround =0}^\round \step_\subround\inner{\bias_\subround}{\pstate_\subround -\eq}}{\tau_\round}}\\
     &+\tau_\round\parens{\dfrac{\sum_{\subround =0}^\round \step_\subround\inner{\noise_\subround}{\pstate_\subround -\eq}}{\tau_\round} + \dfrac{\frac{1}{2\stcon}\sum_{\subround =0}^\round \step_\subround^2\norm{\npayv_\subround}_*^2}{\tau_\round}}
\end{aligned}
\end{equation}
We focus on the asymptotic behavior of each particular term of the previous inequality. We remind that $\filter_{\run}$ denotes the history of $\state_{\run}$ up to stage $\run$ (inclusive) and thus the feedback signal, $\npayv_{\run}$ is not $\filter_{\run}$-measurable in general.
\begin{itemize}[topsep=.5ex,leftmargin=\parindent]
\setlength{\itemsep}{0pt}
\setlength{\parskip}{.2ex}
    \item Let $R_\round = \sum_{\subround =0}^\round \step_\subround^2 \norm{\npayv_\subround}^2_*$. Then \begin{equation}\label{eq: Rn lemma}
    \exof{R_\round} \leq \sum_{\subround =0}^\round \step_\subround^2 \exof{\norm{\npayv_\subround}^2_*} = \sum_{\subround =0}^\round \step_\subround^2 \exof{\exof{\norm{\npayv_\subround}_*^2\given\filter_\subround}} 
    \leq \sum_{\subround =0}^\round\step_\subround^2 \sbound_\subround^2 < \infty
\end{equation}
where $\sum_{\subround =0}^\round \step_\subround^2 \sbound_\subround^2$ is finite by assumption \ref{asm:variance}. 
Hence by \cref{fact: submartingale} and \eqref{eq: Rn lemma} $R_\round$ is an $L_1$ bounded submartingale while Doob's convergence theorem (\cref{Doob's convergence}) shows that almost surely
\begin{equation}
    \lim_{\round \to\infty} \tau_\round^{-1}R_\round = 0 
\end{equation}
\item Let $S_\round = \sum_{\subround=0}^\round\step_\subround\inner{\noise_\subround}{\pstate_\subround -\eq}$ and $\psi_\subround =\step_\subround \inner{\noise_\subround}{\pstate_\subround -\eq}$. For the expected value of  $\psi_\round$ we have
\begin{equation}
    \exof{\psi_\round\given\filter_\round} =\step_\round \inner{\exof{\noise_\round\given\filter_\round}}{\pstate_\round -\eq} = 0
\end{equation}
and so $S_\round$ is a martingale since $\exof{S_\round\given\filter_\round} = S_{\round-1}$. Moreover, for the expectation of the absolute value of $\psi_\round$, Cauchy-Schwarz inequality implies
\begin{align}
    \exof{\abs{\psi_\round}^2\given\filter_\round} &\leq \step_\round^2\exof{\norm{\noise_\round}^2_*\norm{\curr -\eq}^2\given\filter_\round}\\
    &\leq \step_\round^2\exof{\norm{\noise_\round}^2_*\given\filter_\round}\norm{\pspace}^2\\
    &\leq \step_\round^2\sbound_\round^2\norm{\pspace}^2
\end{align}
since 
\begin{align}
    \exof{\norm{\noise_\round}_*^2\given\filter_\round} &= \exof{\norm{\npayv_\round - \exof{\npayv_\round\given\filter_\round}}^2_*\given\filter_\round}\\
    &=\exof{\norm{\npayv_\round}^2_* -2\inner{\npayv_\round}{\exof{\npayv_\round\given\filter_\round}} + \norm{\exof{\npayv_\round\given\filter_\round}}_*^2\given\filter_\round}\\
    &=\exof{\norm{\npayv_\round}^2_*\given\filter_\round} - \norm{\exof{\npayv_\round\given\filter_\round}}^2_*\\
    &\leq \exof{\norm{\npayv_\round}^2_*\given\filter_\round} \leq\sbound_\round^2
\end{align}
where $\sbound_\round^2$ is the upper bound of $\exof{\norm{\npayv_\round}^2_*\given\filter_\round}$ described in \cref{sec:feedback}.\\ Obviously,
$\sum_{\round =0}^\infty\tau_\round^{-2} \exof{\abs{\psi_\round}^2\given\filter_\round} <\infty$ and so by the strong law of large number for martingales (\cref{law of large numbers}) yields that almost surely
\begin{equation}
    \lim_{\round\to\infty}\tau_\round^{-1} S_\round = 0 
\end{equation}
\item Let $W_\round =\sum_{\subround=0}^\round \step_\subround\inner{\bias_\subround}{\pstate_\subround -\eq}$ then by Cauchy-Schwarz inequality
\begin{equation}
\begin{aligned}
    \abs{\tau^{-1}_\round W_\round}\leq\abs{\tau_\round^{-1}\sum_{\subround=0}^\round\step_\subround\inner{\bias_\subround}{\pstate_\subround -\eq}}&\leq \tau_\round^{-1}\sum_{\subround=0}^\round\step_\subround\abs{ \inner{\bias_\subround}{\pstate_\subround -\eq}}\\
    &\leq \tau_\round^{-1}\sum_{\subround=0}^\round\step_\subround\norm{\bias_\subround}_*\norm{\pspace}
    \end{aligned}
\end{equation}
Let $J_\round = \sum_{\subround=0}^\round\step_\subround\norm{\bias_\subround}_*\norm{\pspace}$. Notice that $W_\round \leq J_\round$  and that from \cref{fact: submartingale} $J_\round$ is a submartingale with 
\begin{equation}
    \exof{J_\round}  = \norm{\pspace}\sum_{\subround=0}^\round\step_\subround\exof{\norm{\bias_\subround}_*} \leq \norm{\pspace}\sum_{\subround =0}^\round \step_\subround \exof{\exof{\norm{\bias_\subround}_*\given\filter_\subround}}\leq \norm{\pspace}\sum_{\subround=0}^\round \step_\subround \bbound_\subround <\infty
\end{equation}
where $\bbound_\round$ is the upper bound of $\exof{\norm{\bias_\round}_*\given\filter_\round}$. Thus, $J_\round$ is a $L_1$ bounded submartingale and by Doob's convergence theorem (\cref{Doob's convergence}) almost surely 
\begin{equation}
    \lim_{\round\to \infty}\tau_\round^{-1}J_\round = 0
\end{equation}
As a result, $\tau_\round^{-1}W_\round \to 0$.
\item Finally, we will examine the term $\sum_{\subround =0}^\round \step_\subround\inner{\payv\parens{\pstate_\subround}}{\pstate_\subround -\eq}$. Recall that we had assumed that  $\pstate_\round\in \mathcal{R}\setminus \nhd$ for all $\round\geq 0$, while variational stability holds in $\mathcal{R}$, so by continuity there exists $c>0$, such that for all $\round\geq 0$
\begin{equation}
    \inner{\payv\parens{\pstate_\round}}{\pstate_\round -\eq} \leq -c 
\end{equation}
 \end{itemize}
We return to \eqref{fenchel inequality 2} and we equivalently we have that
\begin{equation}
\begin{aligned}
     \fench_\hreg\parens{\eq,\dstate_{\round+1}} &\leq \fench_\hreg\parens{\eq,\dstate_0} + \sum_{\subround =0}^\round \step_\subround\inner{\payv\parens{\pstate_\subround}}{\pstate_\subround - \eq} +\tau_\round\parens{\tau_{\round}^{-1}W_\round+\tau_\round^{-1}R_\round +\tau_{-1}^\round S_\round}\\
     &\leq \fench_\hreg\parens{\eq,\dstate_0} -c\tau_\round +\tau_\round\parens{\tau_{\round}^{-1}W_\round+\tau_\round^{-1}R_\round +\tau^{-1}_\round S_\round}
     \end{aligned}
\end{equation}
Thus, $\fench_\hreg\parens{\eq,\dstate_{\round+1}} \sim -c\sum_{\subround=0}^\infty \step_\subround \to -\infty$.

By \cref{prop:Prop Fenchel} we conclude to a contradiction. This implies that some instance of the sequence of play is included to every neighborhood $\nhd$ of $\eq$ and thus there exists subsequence $\pstate_{\round_\subround}$ of $\pstate_\round$ that almost surely converges to $\eq$.
\end{proof}

\begin{theorem}[Restatement of \cref{theorem:if} ]\label{th:strict are stable}
Let $\eq$ be a strict Nash equilibrium. If \eqref{eq:FTRL} is run with  payoff feedback that satisfies \ref{asm:bias}-\ref{asm:variance}, then $\eq$ is stochastically asymptotically stable. %under the hypothesis $\sum_{\subround =1}^\infty \step_\subround \to \infty$ and $\sum_{\subround =1}^\infty \step_\subround^2 <\infty$.
\end{theorem}
\begin{proof}
%Fix a confidence level $\conlevel$ and let $U_\eps^* = \left\{y\in\dspace: \fench_\hreg\parens{y,\eq} < \eps\right\}$, then  by \cref{Prop Fenchel} for all $x=\mirror\parens{y}$, $y\in U_\eps^*$ it holds that $\norm{x-\eq}^2 < 2\eps/\stcon$. Let now $U_\eps = \braces{x: \breg_\hreg\parens{\eq,x}<\eps}$, by \cref{Prop Bregman} it is $\norm{x-\eq}^2 <2\eps/\stcon$ and by \cref{Prop Fenchel} $\mirror\parens{U_\eps^*}\subseteq U_\eps$, $\mirror^{-1}\parens{U_\eps}= U_\eps^*$. Since by \eqref{Bregman reciprocity} $U_\eps$ is a neighborhood of $\eq$, we conclude that  whenever $y\in U_\eps^*$, $x=\mirror\parens{y}\in U_\eps$. 
Fix a confidence level $\conlevel$ and let $\nhd_\eps = \braces{x: \breg_\hreg\parens{\eq,x}<\eps}$ and $\nhd_\eps^* = \left\{y\in\dspace: \fench_\hreg\parens{\eq,y} < \eps\right\}$.
\begin{itemize}[topsep=.5ex,leftmargin=\parindent]
\setlength{\itemsep}{0pt}
\setlength{\parskip}{.2ex}
\item By \cref{prop:Prop Bregman} for all $x\in \nhd_\eps$ it holds that $\norm{x-\eq}^2 < 2\eps/\stcon$. 
\item By \cref{prop:Prop Fenchel}  for all $x=\mirror\parens{y}$, $y\in \nhd_\eps^*$ it holds that $\norm{x-\eq}^2 < 2\eps/\stcon$. 
\item Notice that from \cref{prop:Prop Fenchel}  $\mirror\parens{\nhd_\eps^*}\subseteq \nhd_\eps$ and $\mirror^{-1}\parens{\nhd_\eps}= \nhd_\eps^*$. 
\end{itemize}
Thus we conclude that whenever $y\in \nhd_\eps^*$, $x=\mirror\parens{y}\in \nhd_\eps$. Finally, by \eqref{eq:Bregman reciprocity} $\nhd_\eps$ is a neighborhood of $\eq$ .
 Since $\eq$ is a strict \acl{NE}, pick  $\eps$ sufficiently small such that \eqref{eq:variational stability} holds for all $x\in \nhd_{4\eps}$.\\

\noindent\emph{(Stability).}\\
 Assume now that $\dstate_0\in \nhd_\eps^*$ and thus $\fench_\hreg\parens{\eq,\dstate_0} <\eps \leq 4\eps$. We will prove by induction that $\dstate_n\in \nhd_{4\eps}^*$ for all $\round\geq 1$ with probability at least $1-\conlevel$. Suppose that  $\fench_\hreg\parens{\eq,\dstate_\subround}<4\eps$ for all $1\leq \subround\leq \round$ and we will prove that $\dstate_{\round +1}\in \nhd_{4\eps}^*$ and consequently $\pstate_{\round+1}\in \nhd_{4\eps}$.\\\\
From \cref{prop:Prop Fenchel} we have
\begin{equation}
    \fench_\hreg\parens{\eq,\dstate_{\round+1}} \leq \fench_\hreg\parens{\eq,\dstate_\round} + \step_\round\inner{\npayv_\round}{\curr-\eq} + \dfrac{1}{2\stcon}\step_\round^2\norm{\npayv_\round}^2_*
\end{equation}
For the payoff feedback, it holds $\npayv_\round = \payv\parens{\pstate_\round} + \noise_\round + \bias_\round$. Then by telescoping the above inequality and substituting we get 
\begin{equation}\label{key inequality 3}
\begin{aligned}
    \fench_\hreg\parens{\eq,\dstate_{\round+1}} \leq \fench_\hreg\parens{\eq,\dstate_0} &+ \sum_{\subround =0}^\round \step_\subround\inner{\payv\parens{\pstate_\subround}}{\pstate_\subround-\eq}
    + \sum_{\subround=0}^\round \step_\subround\inner{\noise_\subround}{\pstate_\subround-\eq}\\
    &+ \sum_{\subround=0}^\round \step_\subround \inner{\bias_\subround}{\pstate_\subround -\eq} + \dfrac{1}{2\stcon}\sum_{\subround=0}^\round \step_\subround^2 \norm{\npayv_\subround}^2_*
\end{aligned}
\end{equation}
We will study each term of the inequality separately.
\begin{itemize}[topsep=.5ex,leftmargin=\parindent]
\setlength{\itemsep}{0pt}
\setlength{\parskip}{.2ex}
\item Let  $R_\round = \dfrac{1}{2\stcon}\sum_{\subround =0}^\round \step_\subround^2 \norm{\npayv_\subround}^2_*$  and $F_{\round, \eps} = \left\{\sup_{0\leq\subround\leq\round} R_\subround \geq \eps\right\}$. As we discussed in \cref{lem:strict remain to neighborhood}, $R_\round$  is a submartingale with $\exof{R_\round} \leq \sum_{\subround=0}^\round\step_\subround^2\sbound_\subround^2$. Doob's maximal inequality (\cref{Max inequality submartingales}) yields 
\begin{equation}
    \probof{F_{\round,\eps}}\leq\dfrac{\exof{R_\round}}{\eps}\leq \dfrac{\sum_{\subround=0}^\round\step_\subround^2\sbound_\subround^2}{2\stcon\eps}
\end{equation}
By demanding $\sum_{\subround=0}^\infty \step_\subround^2 \sbound_\subround^2 \leq 2\stcon \eps \conlevel/3$ the event $F_\eps =\union_{\round=0}^\infty F_{\eps,\round}$ will occur with probability at most $\conlevel/3$.\\\\
\item Let $S_\round = \sum_{\subround=0}^\round\step_\subround \inner{\noise_\subround}{\pstate_\subround-\eq}$ and $E_{\round,\eps} = \left\{\sup_{0\leq\subround\leq\round}S_\subround\geq \eps\right\}$. Since  $S_\round$ is a martingale, as we discussed in \cref{lem:strict remain to neighborhood}, Doob's maximal inequality (\cref{Max inequality martingales}) yields
\begin{equation}
    \probof{E_{\round,\eps}}\leq \dfrac{\exof{{S_\round}^2}}{\eps^2}\leq \dfrac{\norm{\pspace}^2\sum_{\subround=0}^{\round}\step_\subround^2\sbound_\subround^2}{\eps^2}
\end{equation}
In order to calculate the above upper bound, we define $\psi_\subround = \inner{\noise_\subround}{\pstate_\subround-\eq} $. Notice that 
$ S_\round^2 = \sum_{\subround=0}^{\round} \abs{\psi_\subround}^2 +2 \sum_{\subround<\ell}^{\round} \psi_\subround \psi_\ell    $.
Indeed it holds that
\begin{align}
    \exof{\abs{\psi_\subround}^2}&\leq \exof{\exof{\norm{\noise_\subround}_*^2\norm{\pstate_\subround-\eq}^2\given\filter_\subround}}\\
    &\leq \exof{\exof{\norm{\noise_\subround}^2_*\given\filter_\subround}}\norm{\pspace}^2
\end{align}
where,
\begin{align}
    \exof{\norm{\noise_\subround}_*^2\given\filter_\subround} &= \exof{\norm{\npayv_\subround - \exof{\npayv_\subround\given\filter_\subround}}^2_*\given\filter_\subround}\\
    &=\exof{\norm{\npayv_\subround}^2_* -2\inner{\npayv_\subround}{\exof{\npayv_\subround\given\filter_\subround}} + \norm{\exof{\npayv_\subround\given\filter_\subround}}_*^2\given\filter_\subround}\\
    &=\exof{\norm{\npayv_\subround}^2_*\given\filter_\subround} - \norm{\exof{\npayv_\subround\given\filter_\subround}}^2_*\leq \sbound_\subround^2\\
    &\leq \exof{\norm{\npayv_\subround}^2_*\given\filter_\subround} \leq\sbound_\subround^2
\end{align}
Furthermore, for all $k\neq \ell$ it holds that
    $\exof{\psi_\subround\psi_\ell} = \exof{\exof{\psi_\subround\psi_\ell\given\filter_{k\vee \ell}}}= 0$.\\
Thus, by demanding $\sum_{\subround=0}^\infty \step_\subround^2\sbound_\subround^2 \leq \dfrac{\eps^2\conlevel}{3\norm{\pspace}^2}$ we ensure that the event $E_{\eps} = \union_{\round=0}^\infty E_{\eps,\round}$ will occur with probability at most $\conlevel/3$.\\
\item Let $W_\round = \sum_{\subround=0}^\round\step_\subround \inner{\bias_\subround}{\pstate_\subround - \eq}$, $J_\round = \sum_{\subround =0}^\round \step_\subround \norm{\bias_\subround}_*\norm{\pspace}$ as we discussed in \cref{lem:strict remain to neighborhood}
\begin{equation}
    W_\round\leq J_\round
\end{equation}%then
%\begin{equation}
%   \abs{\sum_{\subround=1}^\round\step_\subround \inner{\bias_\subround}{\pstate_\subround - \eq}}\leq \sum_{\subround =1}^\round \step_\subround \abs{\inner{\bias_\subround}{\pstate_\subround - \eq}} \leq \sum_{\subround =1}^\round \step_\subround \norm{\bias_\subround}_*\norm{\pspace}
%\end{equation}
%Let $J_\round = \sum_{\subround =1}^\round \step_\subround \norm{\bias_\subround}_*\norm{\pspace}$ 
where $J_\round$ is a submartingale with $\exof{J_\round}\leq \norm{\pspace}\sum_{\subround =0}^\round \step_\subround \bbound_\subround$. Similarly to the previous steps let $D_{\eps,\round} = \braces{\sup_{0\leq\subround\leq\round}J_\subround\geq\eps }$, then Doob's maximal inequality (\cref{Max inequality submartingales}) yields
\begin{equation}
    \probof{D_{\eps,\round}}\leq\dfrac{\exof{J_\round}}{\eps}\leq \dfrac{\norm{\pspace}\sum_{\subround =0}^\round \step_\subround \bbound_\subround}{\eps}
\end{equation}
By demanding $\sum_{\subround =0}^\infty \step_\subround \bbound_\subround\leq \dfrac{\eps\conlevel}{3\norm{\pspace}}$ then the event $D_\eps = \cup_{\round =0}^\infty D_{\eps,\round}$ will happen with probability at most $\conlevel/3$, which implies that with probability at most $\conlevel/3$ $W_\round$ will exceed $\eps$ for all $\round\geq 0$.\\
\item Furthermore, if  $\pstate_\subround$ belongs to a neighborhood in which \eqref{eq:variational stability} holds for all $0\leq\subround\leq\round$, we have
\begin{equation}
    \inner{\payv\parens{\pstate_\subround}}{\pstate_\subround-\eq}\leq 0 \text{ for all } \round\geq 0
\end{equation}
\end{itemize}
\noindent By demanding the parameters of the algorithm to satisfy:
\[ \sum_{\subround=0}^\infty\step_\subround^2\sbound_\subround^2\leq \min\left\{\dfrac{\eps^2\conlevel}{3\norm{\pspace}^2},\dfrac{2\stcon \eps \conlevel}{3}\right\}\ \ \  \& \ \ \ \sum_{\subround=0}^\infty\step_\subround\bbound_\subround \leq \dfrac{\eps\conlevel}{3\norm{\pspace}\norm{\dspace}_*}
\]
\noindent If all of $\bar{E}_\eps,\bar{F}_\eps,\bar{D}_\eps$ hold, this happens with probability $\probof{\bar{E}_\eps\intersect\bar{F}_\eps\intersect\bar{D}_\eps}\geq 1-\conlevel$ and from \eqref{key inequality 3} we have $\fench_\hreg\parens{\eq,\dstate_{\round+1}}<4\eps$.
This immediately yields that $\dstate_{\round+1}\in \nhd_{4\eps}^*$ and consequently as we explained in the begin of the proof $\pstate_{\round+1}\in \nhd_{4\eps}$, in which variational stability holds, with probability at least $1-\conlevel$. \\

\noindent\emph{(Convergence).}\\
By \cref{lem:strict remain to neighborhood} there exists a subsequence $\pstate_{\round_\subround}$ that converges to $\eq$. By \eqref{eq: Fenchel reciprocity} we have that $\liminf_{\round\to\infty}{\fench_\hreg\parens{\eq,\dstate_\round}} = 0$. In order to complete the proof, it is sufficient to prove that the limit of $\fench_\hreg\parens{\eq,\dstate_\round}$ exists.
Notice that since the sequence of play remains in $\nhd_{4\eps}$ variational stability holds and thus $\inner{\payv\parens{\curr[\pstate]}}{\curr[\pstate]-\eq}\leq 0$. Again using \cref{prop:Prop Fenchel} we have:
\begin{equation}
    \fench_\hreg\parens{\eq,\dstate_{\round+1}}\leq\fench_\hreg\parens{\eq,\dstate_\round} + \step_\round\inner{\npayv_\round}{\curr-\eq} +\dfrac{1}{2\stcon}\step_\round^2\norm{\npayv_\round}_*^2
\end{equation}
\begin{align}
    \exof{\fench_\hreg\parens{\eq,\dstate_{\round+1}}\given\filter_\round}&\leq \fench_\hreg\parens{\eq,\dstate_\round} + \step_\round\exof{\inner{\bias_\round}{\pstate_\round-\eq}\given\filter_\round} +\dfrac{1}{2\stcon}\step_\round^2\exof{\norm{\npayv_\round}_*^2\given\filter_\round}\\
    &\leq \fench_\hreg\parens{\eq,\dstate_\round} + \step_\round\exof{\inner{\bias_\round}{\pstate_\round-\eq}\given\filter_\round} + \dfrac{1}{2\stcon}\step_\round^2\sbound_\round^2\label{above equation}
\end{align}
Notice that since from \cref{prop:Prop Fenchel} $\fench_\hreg\parens{\eq,\dstate} \geq 0$, if we apply absolute values in the above inequality we have
\begin{align}
 \exof{\fench_\hreg\parens{\eq,\dstate_{\round+1}}\given\filter_\round} &= \abs{\exof{\fench_\hreg\parens{\eq,\dstate_{\round+1}}\given\filter_\round}}\\&\leq \abs{\fench_\hreg\parens{\eq,\dstate_\round}} + \step_\round\exof{\abs{\inner{\bias_\round}{\pstate_\round-\eq}}\given\filter_\round} +\dfrac{1}{2\stcon}\step_\round^2\sbound_\round^2\\
    &\leq \fench_\hreg\parens{\eq,\dstate_\round} + \step_\round\exof{\norm{\bias_\round}_*\given\filter_\round}\norm{\pspace} + \dfrac{1}{2\stcon}\step_\round^2\sbound_\round^2\\
    &\leq \fench_\hreg\parens{\eq,\dstate_\round} + \step_\round\bbound_\round\norm{\pspace} + \dfrac{1}{2\stcon}\step_\round^2\sbound_\round^2
\end{align}
Let
\begin{equation}
    R_\round = \fench_\hreg\parens{\eq,\dstate_\round} + \norm{\pspace}\sum_{\subround = \round}^\infty\step_\subround \bbound_\subround + \dfrac{1}{2\stcon}\sum_{\subround =\round}^\infty\step_\subround^2\sbound_\subround^2
    %\abs{\fench_\hreg\parens{\dstate_\round,\eq} + \sum_{\subround = \round}^\infty\step_\subround\inner{\bias_\subround}{\pstate_\subround-\eq} + \dfrac{1}{2\stcon}\sum_{\subround =\round}^\infty\step_\subround^2\sbound_\subround^2}\\
    %&\leq \abs{\fench_\hreg\parens{\dstate_\round,\eq}}+ \abs{\sum_{\subround = \round}^\infty\step_\subround\inner{\bias_\subround}{\pstate_\subround-\eq}} + \dfrac{1}{2\stcon}\sum_{\subround =\round}^\infty\step_\subround^2\sbound_\subround^2\\
    %&\leq \abs{\fench_\hreg\parens{\dstate_\round,\eq}} + \sum_{\subround = \round}^\infty\step_\subround\norm{\bias_\subround}_*\norm{\pspace} + \dfrac{1}{2\stcon}\sum_{\subround =\round}^\infty\step_\subround^2\sbound_\subround^2
\end{equation}
Then 
\begin{align}
    \exof{R_{\round +1}\given\filter_\round}&\leq \exof{\fench_\hreg\parens{\eq,\dstate_{\round+1}}\given\filter_{\round}} + \sum_{\subround = \round+1}^\infty\step_\subround\bbound_\subround\norm{\pspace} + \dfrac{1}{2\stcon}\sum_{\subround =\round+1}^\infty\step_\subround^2\sbound_\subround^2\\
    &\leq\fench_\hreg\parens{\eq,\dstate_\round} + \sum_{\subround = \round}^\infty\step_\subround\bbound_\subround\norm{\pspace} + \dfrac{1}{2\stcon}\sum_{\subround =\round}^\infty\step_\subround^2\sbound_\subround^2\\
    &= R_\round
\end{align}
Therefore $R_\round$ is a  supermartingale and it is also $L_1$ bounded (each one of the terms is bounded) and so from Doob's convergence theorem (\cref{Doob's convergence}) $R_\round$ converges to a finite random variable and so does $\fench_\hreg\parens{\eq,\dstate_\round}$. Inevitably, $\liminf_{\round\to\infty}\fench_\hreg\parens{\eq,\dstate_\round} = \lim_{\round\to\infty} \fench_\hreg\parens{\eq,\dstate_\round} = 0$ and by \cref{prop:Prop Fenchel}, $\mirror\parens{\dstate_\round} = \curr \to \eq$.

The above analysis shows that whenever $\dstate_0\in \nhd_{\eps}^*$ and thus $\pstate_0\in \nhd_\eps\cap \im\mirror$, $\curr\in \nhd_{4\eps}\cap \im\mirror$ and converges to $\eq$ with arbitrary high probability. Hence, $\eq$ is stochastically asymptotically stable.

\end{proof}

%----------------------------------------------------------------------
%%% APP: INSTABILITY
%----------------------------------------------------------------------
\section{Proof of instability of mixed Nash equilibria}
\label{appendix instability}
\subsection{Proofs of assumptions for \cref{bandit feedback}, \cref{semibandit feedback}}
Below we provide a proof for our claim in  \cref{bandit feedback} that $\bias_\round = \bigoh\parens{\epar_\round}$, $\sbound_\round^2 =\bigoh\parens{1/\epar_\round}$. Focusing on one player $\play\in\players$, notice that
\begin{equation}\label{eq:mean bandit}
    \exof{\npayv_{\play,\round}\given\filter_\round} = \sum_{\pure_{-\play}\in\pures_{-\play}}\estate_{-\play,\round} \parens{\pay_\play\parens{\pure_{\play,1};\pure_{-\play}},\hdots,\pay_\play\parens{\pure_{\play,\abs{\pures_\play}};\pure_{-\play}}} = \payv_\play\parens{\estate_\round}
    \end{equation}
Having this in mind $\npayv_{\play,\round}$ can be viewed as
\begin{equation}
    \npayv_{\play,\round} = \payv_{\play}\parens{\pstate_\round}+\noise_{\play,\round} 
    +\bias_{\play,\round}
\end{equation}
where $\noise_{\play,\round} = \npayv_{\play,\round} - \exof{\npayv_{\play,\round}\given\filter_\round} = \npayv_{\play,\round} - \payv_\play\parens{\estate_\round}$ and $\bias_{\play,\round} = \payv_\play\parens{\estate_\round} - \payv_\play\parens{\pstate_\round}$. Thus, since $\payv_\play\parens{x}$ is multi-linear in $x$ and $\estate_{\play,\round} = \parens{1-\epar_\round}\pstate_{\play,\round} + \epar_\round/\abs{\pures_\play}$ it follows that $\bias_\round = \bigoh\parens{\epar_\round}$. Finally, similarly to \eqref{eq:mean bandit} we can conclude that $\sbound_\round^2 =\bigoh\parens{1/\epar_\round}$.\\
We continue by proving that assumption \ref{asm:degen} is indeed satisfied for both \cref{semibandit feedback}, \cref{bandit feedback}. This is due to the genericity of the game. Actually in the following lemma and corollaries we show that there exist player $\play\in\players$, strategies $a,b\in\supp\parens{\eq_\play}$ and pure strategy profile $\pure_{-\play}\in\supp\parens{\eq_{-\play}}$, where $\eq$ is a mixed Nash equilibrium such that $\abs{\pay_\play\parens{a;\pure_{-\play}}-\pay_\play\parens{b;\pure_{-\play}}}\geq \noiselb$ for some $\noiselb > 0$. %
In order to acquire the exact statement of \ref{asm:degen}, we have to take into account the round in which the game is evolved. Let $\round >0$ be this round, then when examining the stochastic asymptotic stability of a mixed Nash equilibrium $\eq$, the sequence of play is contained in a neighborhood of $\eq$ and thus all of the strategies belonging to the support of $\eq$ have strictly positive probability to be chosen, verifying the statement of \ref{asm:degen}.
\begin{lemma}\label{lem:not equal payoffs}
If the game is generic and has a mixed Nash equilibrium $\eq$, then there exist player $\play\in\players$, pure strategies $a,b\in\supp\parens{\eq_\play}$ ($a\neq b$) and pure strategy profile $\pure_{-\play}\in\supp\parens{\eq_{-\play}}$ such that $\pay_\play\parens{a;\pure_{-\play}} \neq \pay_\play\parens{b;\pure_{-\play}}$.
\end{lemma}
\begin{proof}
Assume that for all players $\play\in\players$, pure strategy profiles $\pure_{-\play}\in\supp\parens{\eq_{-\play}}$ and pure strategies $a,b\in\supp\parens{\eq_\play}$ it is
\begin{equation}
    \pay_\play\parens{a;\pure_{-\play}} = \pay_\play\parens{b;\pure_{-\play}}
\end{equation}
Then for each player $\play$, this implies that all of the payoffs corresponding to pure strategy profiles, which consists of the support of the equilibrium, are equal. Then each pure strategy profile $\parens{\pure_\play;\pure_{-\play}}\in\supp\parens{\eq}$ is a pure Nash equilibrium, which is a contradiction to the genericity of the game.
\end{proof}
\noindent Immediate implications of \cref{lem:not equal payoffs} are:
\begin{corollary}\label{cor:non zero payoff}
There exists player $\play\in\players$ and pure strategy profile $\parens{\pure_\play;\pure_{-\play}}\in\supp\parens{\eq}$, such that $\pay_\play\parens{\pure_\play;\pure_{-\play}}\neq 0$.
\end{corollary}
\begin{corollary}\label{cor:payoffs bounded}
There exist $\paybound >0$, player $\play$, strategies $a,b\in\supp\parens{\eq_\play}$ and pure strategy profile $\pure_{-\play}\in\supp\parens{\eq_{-\play}}$ such that $\abs{\pay_\play\parens{a;\pure_{-\play}} -\pay_\play\parens{b;\pure_{-\play}}}\geq \paybound$. There also exist $\paybound' >0$ and $\parens{\pure_\play;\pure_{-\play}}\in \supp\parens{\eq}$ such that $\abs{\pay_\play\parens{\pure_\play;\pure_{-\play}}}\geq\paybound'$.
\end{corollary}
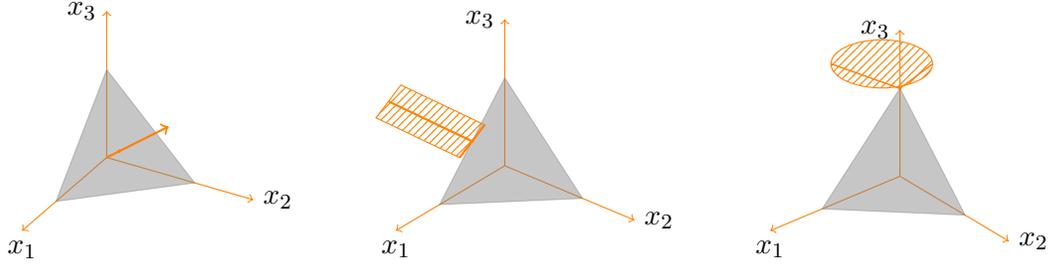
\begin{figure}[!ht]
 \usetikzlibrary{3d}
\tdplotsetmaincoords{60}{120}
\begin{minipage}[t]{.32\textwidth}
\begin{tikzpicture}[tdplot_main_coords,scale = 0.45]
    \def\laxis{5}
    \def\ltriangle{3}
  %  \tdplotsetrotatedcoords{0}{0}{0}
    \begin{scope}[->,orange,tdplot_main_coords]
        \draw (0,0,0) -- (\laxis,0,0) node [below] {\textcolor{black}{$x_1$}};
        \draw (0,0,0) -- (0,\laxis,0) node [right] {\textcolor{black}{$x_2$}};
        \draw (0,0,0) -- (0,0,\laxis) node [left] {\textcolor{black}{$x_3$}};
    \end{scope}
    \filldraw [opacity=0.45,gray,tdplot_main_coords] (\ltriangle,0,0) -- (0,\ltriangle,0) -- (0,0,\ltriangle) -- cycle;
    \filldraw[orange, tdplot_main_coords] (1/3*3,1/3*3,1/3*3) circle (1pt);
    \draw[->,thick,orange,tdplot_main_coords] (0*3,0*3,0*3) -- (1/3*3,1/3*3,1/3*3)-- (5,5,5);
\end{tikzpicture}
\end{minipage}
\begin{minipage}[t]{.32\textwidth}
\begin{tikzpicture}[tdplot_main_coords,scale = 0.45]
    \def\laxis{5}
    \def\ltriangle{3}
    \tdplotsetrotatedcoords{0}{0}{-10}
    \begin{scope}[->,orange,tdplot_rotated_coords]
        \draw (0,0,0) -- (\laxis,0,0) node [below] {\textcolor{black}{$x_1$}};
        \draw (0,0,0) -- (0,\laxis,0) node [right] {\textcolor{black}{$x_2$}};
        \draw (0,0,0) -- (0,0,\laxis) node [left] {\textcolor{black}{$x_3$}};
    \end{scope}
    \def\l{4}
    \def\m{3}
    \filldraw [opacity=.45,gray,tdplot_rotated_coords] (\ltriangle,0,0) -- (0,\ltriangle,0) -- (0,0,\ltriangle) -- cycle;
    \filldraw[orange, tdplot_rotated_coords] (0.5*3,0,0.5*3) circle (1pt);
    %\draw[thick,orange,tdplot_rotated_coords] (1/3-\m,1/3-\m,1/3-\m) -- (1/3+\m,1/3+\m,1/3+\m);
    \draw[thick,orange,tdplot_rotated_coords] (0.5*3-\m,0*3-\m,0.5*3-\m) -- (0.5*3+\m,0*3+\m,0.5*3+\m);
    \draw[thick,orange,tdplot_rotated_coords] (0.5*3,0,0.5*3) -- (0.5*3+1/6*\l,0*3-2/3*\l,0.5*3+1/6*\l);
    \filldraw[orange, tdplot_rotated_coords , pattern=north east lines,pattern color=orange] (0.5*3-\m,0*3-\m,0.5*3-\m) -- (0.5*3+1/6*\l-\m,0*3-2/3*\l-\m,0.5*3+1/6*\l-\m) -- (0.5*3+1/6*\l +\m,0*3-2/3*\l+\m,0.5*3+1/6*\l+\m) -- (0.5*3+\m,0*3+\m,0.5*3+\m) -- cycle ;
\end{tikzpicture}
\end{minipage}
\begin{minipage}[t]{.32\textwidth}

\begin{tikzpicture}[tdplot_main_coords,scale = 0.45]
    \def\laxis{5}
    \def\ltriangle{3}
    \def\a{0.6*\sqrt(2)}
    \tdplotsetrotatedcoords{0}{0}{-20}
    \begin{scope}[->,orange,tdplot_rotated_coords]
        \draw (0,0,0) -- (\laxis,0,0) node [below] {\textcolor{black}{$x_1$}};
        \draw (0,0,0) -- (0,\laxis,0) node [right] {\textcolor{black}{$x_2$}};
        \draw (0,0,0) -- (0,0,\laxis) node [left] {\textcolor{black}{$x_3$}};
    \end{scope}
    \filldraw [opacity=.45,gray,tdplot_rotated_coords] (\ltriangle,0,0) -- (0,\ltriangle,0) -- (0,0,\ltriangle) -- cycle;
    \filldraw[orange, tdplot_main_coords] (0,0,1*3) circle (1pt);
    \filldraw [orange, tdplot_rotated_coords,pattern=north east lines,pattern color=orange] (0,3*0.5,3*1.5) arc (0:180:1.5cm and 0.712cm) -- (0,0,3) -- cycle;
     \filldraw [orange, tdplot_rotated_coords,pattern=north east lines,pattern color=orange] (0,3*0.5,3*1.5) arc (0:-180:1.5cm and 0.712cm) -- (0,0,3) -- cycle;
    %\filldraw[orange, tdplot_main_coords] (0,0,3*1.5) circle (1pt);
    %ellipse (1.5cm and 0.712 cm)
\end{tikzpicture}
\end{minipage}
    \caption{The polar cone corresponding to different points of the simplex. For an interior point this is a line perpendicular to the simplex. For a point of the boundary, it is a plane perpendicular to the simplex tangential to the point of the boundary. For an edge the polar cone corresponds to a cone.}
    \label{fig: polar cone}
\end{figure}

\subsection{Deferred proof of
 \cref{theorem:only-if}}
Before moving on our proof we first provide some intuition derived from the notion of the polar cone (\cref{ss:polar cone}).
Looking at \cref{fig: polar cone}, the polar cone corresponding to \emph{fully mixed} or \emph{mixed} Nash equilibria has a key difference with the one corresponding to \emph{strict} Nash equilibria. The latter, in contrast to the former, is fully dimensional. Thus intuitively, considering a sufficiently small neighborhood of a \emph{mixed} Nash equilibrium, the  slightest perturbation in the dual space of the payoffs, will lead to instability of the system. Our result is based on this intuition; we prove  by contradiction that there exists a sufficiently small neighborhood of a \emph{mixed} Nash equilibrium, from which the sequence of play will escape with strictly positive probability. The decomposability assumption of the regularizers ensures that the proof holds also for steep regularizers (See \cref{ss:steep vs non-steep}). 
Below, leveraging the definition of the polar cone in simplex, we prove a useful property for the difference of the aggregated payoffs of FTRL for a sequence of play that shares common pure strategies.

\begin{lemma}\label{lem:Useful Expression}
Let $\pstate_{\play} = \mirror\parens{\dstate_\play}\in \pspace_\play$ be a mixed strategy profile and $a,b\in\supp\parens{\pstate_\play}$ be two pure strategies, for some player $\play\in\players$. Then it holds:
\begin{equation*}
    \inner{\dstate_\play}{\bvec_a-\bvec_b} = \inner{\nabla \hreg_\play\parens{\pstate_\play}}{\bvec_a-\bvec_b}
\end{equation*}
Additionally, if \eqref{eq:FTRL} is run then for a sequence of play $\pstate_{\play,\round_1},\hdots,\pstate_{\play,\round_2}$ that maintains in its support both pure strategies $a,b\in\pures_\play$ it holds
\begin{equation*}
    \inner{\dstate_{\play,\subround_1} -\dstate_{\play,\subround_2}}{\bvec_a-\bvec_b} = \inner{\nabla \hreg_{\play}\parens{\pstate_{\play,\subround_1}} -  \nabla\hreg_{\play}\parens{\pstate_{\play,\subround_2}}}{\bvec_a-\bvec_b} \ \forall \subround_1,\subround_2\in \{\round_1,\hdots,\round_2\}
\end{equation*}
\begin{comment}
Let $\eq\in\pspace$ be a mixed Nash equilibrium and pure strategies $a,b\in\supp\parens{\eq_\play}$ for some player $\play\in\players$. If \eqref{eq:FTRL} is run and the sequence of play,  is contained in a neighborhood of $\eq$, then for the corresponding aggregate payoffs $\currplay[\dstate],\nextplay[\dstate]$ we have:
\begin{equation*}
    \inner{\nextplay[\dstate] -\currplay[\dstate]}{\bvec_a-\bvec_b} = \inner{\nabla \hreg_{\play}\parens{\nextplay} -  \nabla\hreg_{\play}\parens{\currplay}}{\bvec_a-\bvec_b} 
\end{equation*}
\end{comment}
\end{lemma}

\begin{proof}
From \cref{prop:y in subgradient}, $\dstate_\play$ can be analyzed as $\dstate_\play = \nabla\hreg_\play\parens{\pstate_\play}+ \PCsym$, $\PCsym\in \pcone\parens{\pstate_\play}$. Notice that $\nabla\hreg_\play\parens{\pstate_\play} = \parens{\theta_\play\parens{\pstate_{\play,\pure_1}},\hdots,\theta_\play\parens{\pstate_{\play,\pure_{\abs{\pures_\play}}}}}$. Since $\pstate_\play$ assigns positive probability to both $a,b$, by definition of the polar cone it is $\PCsym_{a} = \PCsym_{b}$. Thus, 
\begin{align}
    \inner{\dstate_\play}{\bvec_a-\bvec_b} &= \PCsym_a + {\theta}'_\play\parens{\pstate_{\play,a}} - \PCsym_b - {\theta}'_\play\parens{\pstate_{\play,b}}\\
    &= \inner{\nabla\hreg_\play\parens{\pstate_\play}}{\bvec_a-\bvec_b}\label{eq1 lemma3}
\end{align}
For the second part, by applying \eqref{eq1 lemma3} for both cases of $\dstate_{\play,\subround_1},\dstate_{\play,\subround_2}$ we have:
\begin{align}
     \inner{\dstate_{\play,\subround_1}}{\bvec_a-\bvec_b} &= \inner{\nabla\hreg_\play\parens{\pstate_{\play,\subround_1}}}{\bvec_a-\bvec_b}\\
     \inner{\dstate_{\play,\subround_2}}{\bvec_a-\bvec_b} &= \inner{\nabla\hreg_\play\parens{\pstate_{\play,\subround_2}}}{\bvec_a-\bvec_b}
\end{align}
From the subtraction of the above equations, we derive the desideratum: 
\begin{equation}
    \inner{\dstate_{\play,\subround_1}-\dstate_{\play,\subround_2}}{\bvec_a-\bvec_b} = \inner{\nabla\hreg_\play\parens{\pstate_{\play,\subround_1}}-\nabla\hreg_\play\parens{\pstate_{\play,\subround_2}}}{\bvec_a-\bvec_b}
\end{equation}

\end{proof}
\begin{theorem}[Restatement of \cref{theorem:only-if} ]\label{th:theorem 1}
Let $\eq$ be a mixed Nash equilibrium. If \eqref{eq:FTRL} is run with any feedback model that satisfies \ref{asm:degen}, then $\eq$ cannot be stochastically asymptotically stable for any choice of step-schedules.
\end{theorem}

\begin{proof}
We start by determining all the parameters of the algorithm \eqref{eq:FTRL} and we assume ad absurdum that $\eq$ is a mixed Nash equilibrium, which is stochastically asymptotically stable. Then for all neighborhoods $\nhd$ of $\eq$ and $\conlevel >0$, there exists some neighborhood $\nhd_{\start}$ such that whenever $\pstate_0 \in \nhd_{\start}$, it holds that $\pstate_\round\in \nhd$ for all $\round\geq 0$ with probability at least $1-\conlevel$. This equivalently implies that for all $\eps,\conlevel >0$ if $\pstate_0\in \nhd_{\start}$, $\norm{\curr -\eq} <\eps$ for all $\round \geq 0$, with probability at least $1-\conlevel$. We leave $\eps$ to be chosen at the end of our analysis, but we will consider it to be fixed.\\
For each player $\play\in\players$ and round $\round$ if $\pstate_{\play,\round},\pstate_{\play,\round+1}$ are two consecutive instances of the sequence of play; then $\norm{\currplay-\eq_\play}<\eps$, $\norm{\nextplay-\eq_\play}<\eps$ and by the triangle inequality
\begin{equation}
    \norm{\nextplay -\currplay} <2\eps
\end{equation}
We fix a round $\round$ and focus on player $\play\in\players$ who has the property of \ref{asm:degen}; Since for two pure  strategies $a,b\in\supp\parens{\eq_\play}$ of player $\play\in\players$, holds that $\probof{\abs{\npayv_{\play a,\round}-\npayv_{\play b,\round}}\geq \noiselb\given\filter_\round}>0$ for all $\round \geq 0$, there exists for each round $\round \geq 0$,  $\problb_\round >0$ such that $\probof{\abs{\npayv_{\play a,\round}-\npayv_{\play b,\round}}\geq \noiselb\given\filter_\round} = \problb_\round$. Choose $\conlevel$ such that $\conlevel < \problb_\round$ and consequently 
\begin{equation}
    1-\conlevel > 1- \problb_\round
\end{equation}
This is possible, since $\problb_\round$ is strictly positive and $\conlevel$ can be chosen arbitrarily small.\\
Consider now the projection of the aggregate payoffs $\currplay[\dstate],\nextplay[\dstate]$ in the difference of the directions of these two strategies. From \cref{lem:Useful Expression} we have
\begin{equation}
    \inner{\nextplay[\dstate]-\currplay[\dstate]}{\bvec_a-\bvec_b} = \inner{\nabla\hreg_\play\parens{\nextplay}-\nabla\hreg_\play\parens{\currplay}}{\bvec_a-\bvec_b}
\end{equation}
However, by definition of \eqref{eq:FTRL}  $\nextplay[\dstate] - \currplay[\dstate] = \step_\round \currplay[\npayv]$ and by taking into consideration that the regularizers used are decomposable, we get
\begin{equation}
    \parens{{\theta}'_\play\parens{\pstate_{\play a,\round+1 }} - {\theta}'_\play\parens{\pstate_{\play b,\round+1}}- \parens{{\theta}'_\play\parens{\pstate_{\play a,\round}} - {\theta}'_\play\parens{\pstate_{\play b,\round}}}}= \step_\round\inner{\currplay[\npayv]}{\bvec_a-\bvec_b}
\end{equation}
By rearranging we have
\begin{equation}\label{key}
   \parensnew{{\theta}'_\play\parens{\pstate_{\play a,\round+1 }} - {\theta}'_\play\parens{\pstate_{\play a,\round }}} - \parensnew{{\theta}'_\play\parens{\pstate_{\play b,\round+1}} - {\theta}'_\play\parens{\pstate_{\play b,\round}}}  = \step_\round\parens{\npayv_{\play a,\round}-\npayv_{\play b,\round}} 
\end{equation}
As a consequence of $\theta_\play$ being continuously differentiable in all of $(0,1]$, $\theta'_\play$ is continuous in $[L\parens{\eps},1]$, where $L\parens{\eps}$ is the lower bound of $\pstate_{\play a},\pstate_{\play b}$ whenever $\norm{\pstate_\play-\eq_\play}<\eps$. $L\parens{\eps}$ can be guaranteed to be positive for a sufficiently small $\eps < \eps'$, which ensures that all the points of the neighborhood contain the support of the equilibrium for player $\play$. Therefore, from extreme value theorem in $\theta'_\play$, there exist finite $\Lipscon_a,\Lipscon_b$ corresponding to $a,b$ equivalently, such that 
\begin{align}
    \absnew{{\theta}'_\play\parens{\pstate_{\play a,\round+1 }} - {\theta}'_\play\parens{\pstate_{\play a,\round}} } &\leq \Lipscon_a\abs{\pstate_{\play a,\round+1} - \pstate_{\play a,\round}} < 2\cdot \Lipscon_a\cdot \eps\label{bounded 1}\\
    \absnew{{\theta}'_\play\parens{\pstate_{\play b,\round+1 }} - {\theta}'_\play\parens{\pstate_{\play b,\round}} }&\leq \Lipscon_b\abs{\pstate_{\play b,\round+1} - \pstate_{\play b,\round}} < 2\cdot \Lipscon_b\cdot \eps\label{bounded 2}
\end{align}
By applying the triangle inequality in \eqref{key} and using \eqref{bounded 1},\eqref{bounded 2} we get
\begin{equation}
    \step_\round \abs{\npayv_{\play a,\round} - \npayv_{\play b,\round}} <
     \parens{2\cdot \Lipscon_a + 2\cdot \Lipscon_b }\cdot \eps
\end{equation}
Equivalently,
\begin{equation}\label{noise dif}
   \abs{\npayv_{\play a,\round} - \npayv_{\play b,\round}}< \dfrac{2\cdot \Lipscon_a + 2\cdot \Lipscon_b } {\step_\round} \cdot \eps
\end{equation}
The above inequality holds with probability $1-\conlevel$. Thus, if the sequence of play $\curr$ is contained to an  $\eps-$neighborhood of $\eq$ \ie $\norm{\pstate_\round -\eq} <\eps$ for all $n\geq 0$, then the difference of the feedback, for some player $\play\in\players$, to two strategies of the equilibrium is $O\parens{\eps/\step_\round}$ with probability at least $1-\conlevel$. \\
We now fix $\eps$ to be
\begin{equation}\label{eps}
    \eps < \min\left\{\eps',\dfrac{\step_\round}{2\cdot \Lipscon_a + 2\cdot \Lipscon_b } \noiselb\right\}
\end{equation}
and consequently
\begin{equation}\label{contradiction part 1}
     \probof{\abs{\npayv_{\play a,\round} - \npayv_{\play b,\round}}< \noiselb\given\filter_\round} \geq 1-\conlevel
\end{equation}
However, from assumption \ref{asm:degen}, it holds that 
\begin{equation}\label{contradiction part 2}
    \probof{\abs{\npayv_{\play a,\round} - \npayv_{\play b,\round}}\geq \noiselb} \geq \problb_\round
\end{equation}
Combining \eqref{contradiction part 1},\eqref{contradiction part 2} we conclude
\begin{align}
    1&=\prob\left[\left\{\abs{\npayv_{\play a,\round} - \npayv_{\play b,\round}} \geq \noiselb \right\}\cup\left\{ \abs{\npayv_{\play a,\round} - \npayv_{\play b,\round}} < \noiselb \right\}\right]\\
    &= \prob\left[\abs{\npayv_{\play a,\round} - \npayv_{\play b,\round}} \geq \noiselb \right] + \prob\left[\abs{\npayv_{\play a,\round} - \npayv_{\play b,\round}} <\noiselb \right]\\
    &\geq \problb_\round +1-\conlevel\\
    &>  1
\end{align}
which is a contradiction.

Thus, a mixed Nash equilibrium cannot be stochastically asymptotically stable, under \eqref{eq:FTRL} for types of payoff feedback described in \cref{sec:feedback}. Notice that this analysis holds even for the first round. Once the parameters of the algorithm have been determined, asymptotic instability can be derived in whichever finite round.

\end{proof}

\end{document}